\documentclass[review]{elsarticle}

\usepackage{lineno,hyperref}
\modulolinenumbers[5]

\journal{Journal of \LaTeX\ Templates}
\usepackage{amsmath, amsthm, amssymb}
\usepackage{graphicx}
\usepackage{amssymb}
\usepackage{epstopdf}
\usepackage{bm}
\usepackage{graphicx,epsfig,color,bm}
\usepackage{booktabs }
\usepackage{threeparttable}
\usepackage{longtable}
 \usepackage{supertabular}
\usepackage{tabularx} 
\usepackage[labelsep=quad]{caption}
\newtheorem{theorem}{Theorem}[section]
\newtheorem{lemma}{Lemma}[section]

\newtheorem{definition}{Definition}[section]
\newtheorem{corollary}{Corollary}[section]
\newtheorem{proposition}{Proposition}[section]

\newcommand{\bI}{{\bf I}}
\newcommand{\bone}{{\bf 1}}

\newcommand{\bX}{{\bf X}}
\newcommand{\bA}{{\bf A}}
\newcommand{\bx}{{\bf x}}
\newcommand{\bxi}{{\bm \xi}}
\newcommand{\bS}{{\bm \Sigma}}
\newcommand{\btea}{{\bm \theta}}
\newcommand{\bmu}{{\bm \mu}}
\newcommand{\mb}{\mbox}
\newcommand{\md}{\mbox{d}}
\newcommand{\mS}{\bf S}

\begin{document}

\begin{frontmatter}

\title{Tests for  large dimensional covariance  structure based on Rao's  score  test
%\runtitle{Large Dimensional Covariance Structure  Test.}
\corref{mycorrespondingauthor}}
\cortext[mycorrespondingauthor]{Corresponding author}
%\tnotetext[mytitlenote]{Fully documented templates are available in the elsarticle package on \href{http://www.ctan.org/tex-archive/macros/latex/contrib/elsarticle}{CTAN}.}

%% Group authors per affiliation:
\author{Dandan Jiang\fnref{myfootnote}}
\address{School of Mathematics, \\
Jilin University, \\
2699  QianJin Street, \\
Changchun {\rm 130012}, China.}
\fntext[myfootnote]{Supported by Project 11471140 from NSFC.}
\ead{jiangdandan@jlu.edu.cn}

%% or include affiliations in footnotes:
\begin{abstract}
This paper proposes  a  new test for covariance matrices  structure based on the correction to Rao's score test in large dimensional framework. By generalizing the  CLT for the linear spectral statistics of large dimensional sample covariance matrices,   the test can be  applicable for large dimensional  non-Gaussian variables in a wider range without  the restriction of the 4th moment. Moreover, the amending Rao's score test  is also powerful  even for the ultra high dimensionality as $p \gg n$, which  breaks the inherent idea that  the corrected tests by RMT can be 
only   used when $p<n$.  Finally, we compare the proposed test with other high dimensional  covariance structure  tests to evaluate their performances through the simulation study.
\end{abstract}


\end{frontmatter}

\linenumbers

%%%%%%%%%%%%%%%%%%%%%%%%%%%%%%%%%%%%%%%%%%%%%%%%%%%
%      Introduction
%%%%%%%%%%%%%%%%%%%%%%%%%%%%%%%%%%%%%%%%%%%%%%%%%%%
\section{Introduction}

Recent advances in data acquisition techniques and the ease of access to high computation power  have  fueled increased interest in analyzing  the data with moderate even large dimensional  variables in most sciences,  such as  microarray gene  expressions in biology, where the number of feature variables $p$  greatly  exceeds  the sample size $n$.  However, the traditional statistical methods encounter failure due to the increase in dimensionality, because they are established on the basis of fixed dimension $p$ as the sample size $n$ tends to infinity.  So 
 many  efforts have been made to improve the power of the classical statistical methods and to propose new procedures designed for  the large dimensional data. 
 A particular attention has been paid to the covariance matrices structure test, which is of  fundamental statistical interest and   widely used in the  biology, finance and  etc.  Let $\chi=(\bx_1,\cdots,\bx_n)$ is an independent and identically distributed  sample from 
a $p$ dimensional random vector $\bX$ with mean $\bmu$ and covariance matrix $\bS$. To test on  the structure of  covariance  matrices, we consider the hypothesis  
\begin{equation}
H_0 : \bS=\bS_0  \quad \mb{v.s.} \quad  H_1: \bS \neq \bS_0, \label{H1}
\end{equation}
which covers  the identity hypothesis test  $H_0 : \bS=\bI_p $ and  the sphericity  hypothesis test  $H_0 : \bS=\gamma\bI_p $ as  the special cases. Within this context,  it has been  well studied under the normal distribution assumption with the classical setting of fixed $p$, such as  \cite{A03}, \cite{John} and \cite{Nagao}.  Also, the Rao's score test was given in  \cite{Rao}.  But they all  lost their effectiveness  as $p$ is  a moderate or ultra high  dimension, even worse for the non-Gaussian variables. Therefore, many statisticians have  investigated this problem and provided  the  various solutions  for  the large dimensional data setting. The earlier  works include \cite{Johnstone},  \cite{Ledoitwolf} and  \cite{Srivastava}, which involved some well-chosen distance function  relied on the first  and second spectral moments as dimension $p$ and sample size $n$ go to infinity together, whereas  they were invalid for either the ultra high dimensionality or non-Gaussian variables. Then Bai  {\em et al}. \cite{B09} focused   on   deriving the limiting behavior of the corrected  LRT under the large dimensional limiting scheme $p/n \rightarrow c \in [0,1)$,  and Jiang  {\em et al}. \cite{J12} extended it to a wider spread with $c \in [0,1]$ and $p<n$. Their methods expanded the application range  without  distribution assumption, but still not applicable for the case of  $ p>n$ where the likelihood ratio cannot be well defined. Recently, Chen {\em et al}. \cite{Chen} proposed  a  nonparametric test  with the constrains of  uniformly  bounded  8th moment  and derived its asymptotic distribution under the null hypothesis  regardless of the limiting  behavior of $p/n$.  Motivated by  this,    Cai and Ma \cite{Cai}  investigated the high dimensional covariance testing problem from a minimax point of view under the normal assumption. It showed that its power uniformly dominated that of the corrected LRTs over the entire asymptotic  regime  in which  the corrected LRTs  were defined. Though it had the optimal power,  as seen from our simulation, It  failed in empirical sizes  when the dimension $p$ was much higher than the sample size $n$, especially  the case of "large $p$ small $n$".

In this paper, we proposed a  new test for  the hypothesis (\ref{H1})  by RMT (random matrix theory) based on the aforementioned  Rao's score test. The main contributions of this work displayed in several aspects. First, we generalized  the   CLT(central limit theorem) for   the  LSS (linear spectral statistic) of large dimensional sample covariance matrices in  \cite{BS04}. By removing the restriction that the 4th moment  of the variable is $3+\delta$, where $\delta $ is a positive constant tending to 0,  we provided  an enhanced  version  of the theorem, which made the test proposed in this work  suitable for non-Gaussian variables in a wider range.  Moreover,  our correction based on Rao's score test can be applied to the ultra high dimensionality  in despite of  the functional relationship between $p$ and $n$. Although it was derived  under the limiting scheme $p/(n-1)\rightarrow q \in [0, +\infty) $  with unknown mean parameter $\bmu$, exactly what we need was just the ratio of $p$ over $n$  in practical problems, which is always easily acquired under any  functional expression  of $p$ and $n$. It can be sustained by the simulation when $(p,n)=(40,19)$ or  $(320,79)$ and etc., which are close to the pair numbers adopted 
in  \cite{Chen}  by the function $p=\mb{exp}(n^{0.4})+10$.  It also revealed that  whether the corrections by RMT can be used  in the case of $p>n$ depends on the corrected statistics we chose  rather than the tools we used in RMT.    Finally,  the restricted condition is relaxed  to the finite 4th moment  compared  with \cite{Chen}, and our correction to Rao's score  test has the more accurate sizes  and  better  powers as shown in the simulation study.

The remainder  of the article  is organized as follows. Section \ref{Pre}  gives   a quick review of the  Rao's  Score  test,  then  details  their testing statistics for  covariance  structure   tests.  An enhanced version of the  large dimensional CLT in \cite{BS04} is also provided in this part.  In Section \ref{New}, we propose  the  new  testing statistics in large dimensional setting  based on the Rao's score test.  Simulation results are presented  to evaluate the performance   of  our test compared with other large dimensional covariance matrices tests in Section \ref{Sim}. Then we draw a conclusion in the  Section \ref{Con}, and the proofs and derivations are listed in the Appendix \ref{app}

%%%%%%%%%%%%%%%%%%%%%%%%%%%%%%%%%%%%%%%%%%%%%%%%%%%
%     Preliminary
%%%%%%%%%%%%%%%%%%%%%%%%%%%%%%%%%%%%%%%%%%%%%%%%%%%

\section{Preliminary} \label{Pre}

We first  give a quick review of the Rao's  Score  test, and derive their classical test statistic for the  hypothesis  (\ref{H1}).  Then the test statistic is refined into something precisely needed in  the amendment process.  An enhanced version of the CLT for   LSS of large dimensional  sample covariance matrices  is also presented, which makes it possible that the modifications of  the  score tests  have a wider use with the 4th moment requirement  excluded. 
 
\subsection{Rao's Score Test} 

Let $\bX$ be  a random variable with population distribution $F_{\bX}(x,\btea)$ and density function 
$f_{\bX}(x,\btea)$, where  $\btea$ is an unknown parameter.  The score vector of $\bX$ is defined as  $U(\bX,\btea)=\displaystyle\frac{\md}{\md \btea}\mb{ln} f_{\bX}(x,\btea)$. Then the  information matrix of $\bX$ is 
\[ I (\bX,\btea)=\mb E(U(\bX,\btea)U'(\bX,\btea))\]
It is well known that the information matrix  can be also calculated by Hessian matrix $H(\bX,\btea)$ as below:
\[I (\bX,\btea)=- \mb E (H(\bX,\btea))=- \mb E (\frac{\md^2}{\md \btea^2}\mb{ln} f_{\bX}(x,\btea))\]

Let $\chi=(\bx_1,\cdots,\bx_n)$  denote a sample from the population distribution  $F_{\bX}(x,\btea)$. 
Then the log-likelihood, the score function and the information matrix of the sample are given by 
$l(\chi,\btea)=\sum\limits_{i=1}^{n}\mb{ln} f(\bx_i,\btea)$, $U(\chi,\btea)=\sum\limits_{i=1}^{n}U(\bx_i,\btea)$ and $I (\chi,\btea)=nI (\bx_1,\btea)$, respectively. 
Then we have  the definition of  Rao's score test statistic as below:
\begin{definition}
 Rao's score test statistic    for the hypothesis $H_0 : \btea=\btea_0  $ is defined by
\[\mb{RST}(\chi, \btea_0)=U'(\chi,\btea_0)I(\chi,\btea_0)^{-1}U(\chi,\btea_0),\]
 where $\btea_0=(\theta_{01},\cdots,\theta_{0p})'$ is a known vector and 
 $ \mb{RST}(\chi, \btea_0) $ tends to a  $\chi^2_p$ limiting distribution
  as $n \rightarrow \infty$ under  $H_0$. (Rao,1948).
\end{definition}

To specify  the Rao's score test statistic for  hypothesis test (\ref{H1}), we suppose the sample $\chi=(\bx_1,\cdots,\bx_n)$ follows a normal distribution with mean parameter $\bmu$ and covariance matrix $\bS$.  Denote 
$\btea=(\bmu', \mb{vec}(\bS)')' $, where $\mb{vec} ( \cdot )$ is the vectorization  operator.   First, the  logarithm  of the density of the sample $\chi$ is written  as  
\[l(\chi,\btea)=-\frac{np}{2}\mb{ln}(2\pi)-\frac{n}{2}\ln|\bS|-\frac{1}{2}\sum\limits_{i=1}^{n}\mb{tr}\left(\bS^{-1}(\bx_i-\bmu)(\bx_i-\bmu)'\right).\] 
By the definition
\[U(\chi,\btea)=\frac{\md}{\md \btea}l(\chi,\btea),\]
where $\frac{\md}{\md \btea}=
\left(
\begin{array}{cc}
\frac{\md}{\md \bmu}   \\
\frac{\md}{\md \mb{vec}(\bS)}
\end{array}
\right)
$ is a $p(p+1) \times 1$ vector, then the score vector  for the sample is 
\begin{equation}
U(\chi, \btea)=:\left(
\begin{array}{c}
U_1 (\chi, \btea) \\
U_2(\chi, \btea)
\end{array}
\right)=\left(
\begin{array}{c}
n\bS^{-1}(\hat\bmu-\bmu)  \\
\displaystyle\frac{n}{2}\mb{vec}(\bS^{-1}(\bA\bS^{-1}-\bI_p))
\end{array}
\right) \label{scorevec}
\end{equation}
where 
\begin{equation}
\hat\bmu=\displaystyle\frac{1}{n}\sum\limits_{i=1}^n\bx_i \quad  \mbox{ and}
 \quad  \bA=\displaystyle\frac{1}{n}\sum\limits_{i=1}^{n}(\bx_i-\bmu)(\bx_i-\bmu)'.\label{muS}
\end{equation}  
Derivations of  (\ref{scorevec}) is specified in  the Appendix \ref{A1}.

Secondly,  the Hessian matrix 
$H(\chi, \btea)=\frac{\md^2}{\md \btea^2}l(\chi,\btea)=: 
\left(
\begin{array}{ccc}
 H_{11} & H_{12}    \\
 H_{21} & H_{22}   
\end{array}
\right)
$, where the part of the parameter $\bS$ is
\begin{eqnarray}
H_{22}&=&\frac{n}{2}\frac{\md \mb{vec}(\bS^{-1}(\bA\bS^{-1}-\bI_p))}{\md \mb{vec}'(\bS)}\label{H22d}\label{H22eq}\\ 
&=&\frac{n}{2}\frac{\md \mb{vec}(\bS^{-1})}{\md \mb{vec}'(\bS)} \frac{\md \mb{vec}(\bS^{-1}\bA\bS^{-1}-\bS^{-1})}{\md  \mb{vec}'(\bS^{-1})}\nonumber \\
&=& -\frac{n}{2} (\bS^{-1} \otimes \bS^{-1} )(\bA\bS^{-1} \otimes \bI_p+\bI_p\otimes \bA\bS^{-1}-\bI_{p^2}). \nonumber
\end{eqnarray}
Details of derivations for (\ref{H22eq}) can be found in  the Appendix \ref{A1}.
Because $I (\bX,\btea)=- \mb E (H(\bX,\btea))$ and $\mb{E} (\bA)=\mb E [(\bX-\bmu)(\bX-\bmu)']=\bS$,  where $\bA$ is defined in (\ref{muS}),  then the information matrix 
 \[I(\chi, \btea)=: 
\left(
\begin{array}{ccc}
 I_{11}(\chi, \btea) & I_{12} (\chi, \btea)   \\
 I_{21}(\chi, \btea) & I_{22} (\chi, \btea)  
\end{array}
\right),
\]where  the part  for $\bS$ is 
\[I_{22}(\chi, \btea)=\frac{n}{2} (\bS^{-1} \otimes \bS^{-1} )\]

If there are no restrictions on $\bmu$,  the parameter $\bmu$ in the 
score vector is replaced by its maximum likelihood estimator $\hat\bmu$. Then the part
of the score vector corresponding to $\bmu$ turns to ${\bf 0}$, and  only the second part of the score vector $U_2(\chi, \btea)$  and the  element $I_{22}(\chi, \btea)$ of the information matrix contribute to the calculation of the 
Rao's score test statistic.(See \cite{Gombay}). Therefore, the Rao's score test statistic  for hypothesis test (\ref{H1}) can be calculated by the expressions of 
$U_2(\chi,  \btea)$  and $I_{22}(\chi, \btea)$, where   $\bmu$  and  $\bA$  are substituted  with sample mean $\hat\bmu$ and  the  sample covariance matrix
\begin{equation}
\hat\bS=\displaystyle\frac{1}{n}\sum\limits_{i=1}^{n}(\bx_i-\hat\bmu)(\bx_i-\hat\bmu)', \label{Sigmahat}
\end{equation}  
respectively.  Also $\bS_0$ is  instead of $\bS$ under the null hypothesis.
Thus, we have 
\begin{proposition}
Rao's score test statistic  for testing $H_0 : \bS=\bS_0 $ with no constrains on $\bmu$ has the 
following form 
\begin{equation}
\mb{RST}(\chi, \bS_0)=\frac{n}{2}\mb{tr}[(\bS_0^{-1}\widehat{\bS}-\mb{\bI}_p)^2] \label{RST1}
\end{equation}
where $\chi=(\bx_1,\cdots,\bx_n)$ is a sample from $N_p(\bmu,\bS)$, and the test statistic 
$\mb{RST}(\chi, \bS_0)$ tends to a $\chi^2$ distribution  with freedom  degree $\displaystyle\frac{p(p+1)}{2}$ under $H_0$ when $n \rightarrow \infty$.
\end{proposition}

\begin{proof}
\begin{eqnarray*}
&&\mb{RST}(\chi, \bS_0)\\
&=&\frac{n}{2}\mb{vec}'(\bS_0^{-1}(\hat\bS\bS_0^{-1}-\bI_p)) 
[\frac{n}{2} (\bS_0^{-1} \otimes \bS_0^{-1} )]^{-1}
\frac{n}{2}\mb{vec}(\bS_0^{-1}(\hat\bS\bS_0^{-1}-\bI_p))\\
&=&\frac{n}{2}\mb{vec}'(\bS_0^{-1}(\hat\bS\bS_0^{-1}-\bI_p)) 
\mb{vec}(\hat \bS-\bS_0)\\
&=&\frac{n}{2}\mb{tr}[(\bS_0^{-1}\widehat{\bS}-\mb{\bI}_p)^2]
\end{eqnarray*}
\end{proof}

For some special cases are listed in the corollaries as following.
\begin{corollary}
Rao's score test statistic  for testing $H_0 : \bS=\bI_p $ with no constrains on $\bmu$ has the 
following form 
$$
\mb{RST}(\chi, \bI_p)=\frac{n}{2}\mb{tr}[(\widehat{\bS}-\mb{\bI}_p)^2]
$$
where $\chi=(\bx_1,\cdots,\bx_n)$ is a sample from $N_p(\bmu,\bS)$, and the test statistic 
$\mb{RST}(\chi, \bI_p)$ tends to a $\chi^2$ distribution  with freedom  degree $\displaystyle\frac{p(p+1)}{2}$ under $H_0$ when $n \rightarrow \infty$.
\end{corollary}

\begin{corollary}
Rao's score test statistic  for testing $H_0 : \bS=\gamma\bI_p $ with no constrains on $\bmu$ has the 
following form 
$$
\mb{RST}(\chi, \gamma\bI_p)=\frac{n}{2}\mb{tr}[(\frac{p}{\mb{tr}(\widehat\bS)}\widehat{\bS}-\mb{\bI}_p)^2]
$$
where $\chi=(\bx_1,\cdots,\bx_n)$ is a sample from $N_p(\bmu,\bS)$, and the test statistic 
$\mb{RST}(\chi, \gamma\bI_p)$ tends to a $\chi^2$ distribution  with freedom  degree $\displaystyle\frac{p(p+1)}{2}-1$ under $H_0$ when $n \rightarrow \infty$.
\end{corollary}
\begin{proof}
Replace the $\bS_0$ by $\hat\gamma\bI_p$ according to (\ref{RST1}), where $\hat\gamma=\displaystyle\frac{\mb{tr}(\widehat\bS)}{p}$ is the maximum likelihood estimator of $\gamma$.
\end{proof}

\subsection{CLT for LSS of a large dimensional sample covariance matrix }

As seen above, the statistics of Rao's score test     for the hypothesis (\ref{H1}) can be encoded by the trace  function  of a  matrix, i.e. a function of the eigenvalues of  some matrix concerned  with the sample covariance matrix. That is exactly what we need in the corrections to the score test for large dimensional cases. Consequently,   a quick survey of the  CLT for  LSS of a large dimensional sample  covariance matrix referred in  \cite{BS04} is   presented below, which is a basic tool for   improvements on  the  classical Rao's score test. Because  the original  version of the theorem has a strict condition  that the 4th moment  of the variable is $3+\delta$, where $\delta $ is a positive constant tending to 0,  so  we derive an enhanced  version excluding this requirement for a more widely usage.  Before quoting, we first introduce some basic concepts and notations. 

Suppose   $ \{\xi_{ki} \in \mathbb{C}, i, k = 1, 2, \cdots \}$ be a
double array of $\mb{i.i.d.}$  random  variables with mean 0 and variance
$1$.
 Then  $(\bxi_1, \cdots,\bxi_{n}) $ is  regarded as an $\mb{i.i.d.}$
sample from some $p$-dimensional distribution with mean ${\bf 0}_p$ and
covariance matrix $\bI_p$, where $\bxi_i = (\xi_{1i}, \xi_{2i},\cdots , \xi_{pi})'$. 
So the sample covariance matrix is
\begin{equation}
\mS_n={1\over{n}}\sum\limits_{i=1}^{n}\bxi_i\bxi_i',\label{Sn}
\end{equation}
where we use conjugate transpose  for the complex variables instead.
For simplicity we use  $F^q, F^{q_n}$  to denote the Mar\v{c}enko-Pastur law of index
 $q~ \mbox{and}  ~ q_n$ respectively, where $q_n=\frac{p}n \rightarrow q \in [0, +\infty)$.
 $F_n^{\mS_n}$ marks the empirical spectral distribution(ESD) of the matrix $\mS_n$,  which is defined as 
 \[
F_n^{\mS_n}(x) = \frac{1}{p}\sum\limits_{i=1}^{p}\textbf{ 1}_{\lambda_i^{\mS_n}
\leq x}, \quad \quad x \in \mathbb{R},
\]
  where $\left(\lambda_i^{\mS_n}\right)$ are the real eigenvalues of  the $p
\times p$ square matrix $\mS_n$.   Define 
\[ \int f (x)\md F_n^{\mS_n}(x)=
\frac{1}{p}\sum\limits_{i=1}^p f(\lambda_i^{\mS_n}),
\]
which is a  so-called linear spectral statistic (LSS) of the random matrix ${\mS_n}$. Based on this,  we consider the empirical process
$G_n : = \{G_n(f)\}$ indexed by $\mathcal{A}$ ,
\begin{equation}
G_n(f)= p\cdot \int_{-\infty}^{+\infty} f(x)\left[F^{\mS_n}_n-
F^{q_n}\right] (\md x), ~~~~~~\quad f \in  \mathcal{A},\label{Gdef}
\end{equation}
 where $\mathcal{U}$ is an open set of the complex plane
  including   $[ a(q),b(q)]$,
  where $a(q)=(1-\sqrt{q})^2$ and $b(q)= (1+\sqrt{q})^2 ]$,
and $\mathcal{A}$   be the set of analytic functions $f :
\mathcal{U} \mapsto \mathbb{C}.$  Actually, the  contours in $\mathcal{U}$ should contain the whole supporting set of the LSD $F^{q}$.
It is  known that if $q \leq 1$, exactly it is $[ a(q),b(q)]$. If $q>1$,  the contours should enclose the whole  support $\{0\}  \cup [ a(q),b(q)]$, because the $F^q$ has a positive mass at the origin at this time.  However,  due to the exact separation theorem in \cite{BS99}, for large enough $p$ and $n$, the discrete mass at the origin  will coincide  with that of $F^{q}$. So we can restrict 
the integral $G_n(f)$ on the contours  only enclosed the continuous part of the LSD  $F^{q}$.

Define
\[\kappa=
\left\{
\begin{array}{cc}
 2, & \mbox{if the~}  \bxi- \mbox{variables are real,\quad\quad} \\
  1,& \mbox{if the~}  \bxi- \mbox{variables are complex.} 
\end{array}
\right.
\]
Then  an  enhanced  version   of Theorem 1.1 in  \cite{BS04} is provided, which will play a fundamental role in next derivations.

\begin{lemma}

Assume: \\
   $ f_1, \cdots ,f_k \in \mathcal{A}$,  $\{\xi_{ij}\}$  are  $\mb {i.i.d.}$
random variables, such that  $\mb{E}\xi_{11}=0,~ \mb{E}{|\xi_{11}|^2}=\kappa-1,
~\mb{E}{|\xi_{11}|}^4 < \infty$  and the $\{\xi_{ij}\}$ satisfy the condition
\[\frac{1}{np}\sum\limits_{ij} \mb{E}|\xi_{ij}|^4I(|\xi_{ij}|\geq \sqrt{n}\eta) \rightarrow 0\]
for any fixed $\eta>0$.  Moreover,
 $\displaystyle\frac{p}{n}=q_n \rightarrow q
\in [0, +\infty) $
as  $n, p \rightarrow \infty$  and $E(\xi_{11}^4)=\beta+\kappa+1,$ where $\beta$ is a  constant.\\[1mm]
Then the random vector $\left(G_n(f_1), \cdots 
,  G_n(f_k) \right)$  forms a tight sequence by the  index $n$, and  
it  weakly converges to a $k$-dimensional Gaussian
vector with mean vector
\begin{eqnarray}
&&\mu(f_j)=-\frac{\kappa-1}{2\pi i} \oint f_j(z) \frac{q\underline{m}^3(z)(1+\underline{m}(z))}{[(1-q)\underline{m}^2(z)+2\underline{m}(z)+1]^2} \md z \label{04mean1}\\
&&\quad -\frac{\beta q }{2 \pi i} \oint f_j(z)\frac{\underline{m}^3(z)}{(1+\underline{m}(z))[(1-q)\underline{m}^2(z)+2\underline{m}(z)+1]} \md z,
\label{04mean2}
\end{eqnarray} 
and covariance
function
\begin{eqnarray}
&&\upsilon\left(f_j,
f_\ell\right)=-\frac{\kappa}{4\pi^2}\oint\oint\frac{f_j(z_1)f_\ell(z_2)}{(\underline{m}(z_1)-\underline{m}(z_2))^2}
\md\underline{m}(z_1)\md \underline{m}(z_2)  \label{04var1}\\
&&\quad-\frac{\beta q}{4\pi^2}\oint\oint\frac{f_j(z_1)f_\ell(z_2)}{(1+\underline{m}(z_1))^2(1+\underline{m}(z_2))^2}
\md\underline{m}(z_1)\md \underline{m}(z_2), \label{04var2}
\end{eqnarray}
where  $ j,\ell \in \{1, \cdots,
k\}$, and  $\underline{m}(z)\equiv m_{\underline{F}^q}(z)$ is the
Stieltjes Transform of ~ $\underline{F}^q\equiv (1-q)I_{[0,
    \infty)}+qF^q$. The contours all contain the support of $F^q$ and 
   non overlapping  in both (\ref{04var1})  and    (\ref{04var2}).\\
\label{CLT}
\end{lemma}

The  proof of the  Lemma \ref{CLT} is detailed in Appendix \ref{A2}.

%%%%%%%%%%%%%%%%%%%%%%%%%%%%%%%%%%%%%%%%%%%%%%%%%%%
%    The Proposed Testing Statistic
%%%%%%%%%%%%%%%%%%%%%%%%%%%%%%%%%%%%%%%%%%%%%%%%%%%

\section{The Proposed Testing Statistics} \label{New}

 In this part,  $\chi=(\bx_1,\cdots,\bx_n)$  remains to be an independent and identically distributed  sample from 
a $p$ dimensional random vector $\bX$ with mean $\bmu$ and covariance matrix $\bS$.   For  testing the hypothesis 
   $H_0 :  \bS  =\bS_0$,  set 
$\widetilde\bxi_i=\bS_0^{-\frac{1}{2}}(\textbf{x}_i-\bmu),$
then the array $\{\widetilde\bxi_i\}_{i=1, \cdots, n}$ contains  $p$-dimensional
standardized variables under $H_0$.  If the mean parameter $\bmu$ is known,  the Lemma \ref{CLT} can be cited in a direct way because its  sample covariance matrix is  identical with $\mS_n$ in (\ref{Sn}). However, it shows  a slightly difference   with unknown $\bmu$. By \cite{ZhengBaiYao},  it is reasonable to use $n-1$  instead of $n$,  if applying the CLT in the Lemma \ref{CLT} to correct the score test in large dimensional  data with the unknown  mean parameter.  Also, the estimator of  covariance matrix  in the the corrected statistics should be changed into the unbiased one. 
Therefore, we define the unbiased  sample covariance matrix of $\{\widetilde\bxi_i\}$ as $\widetilde\bS=\bS_0^{-\frac{1}{2}}(\displaystyle\frac{n}{n-1}\widehat\bS)\bS_0^{-\frac{1}{2}}$, and  denote $\mS=\bS_0^{-1}(\displaystyle\frac{n}{n-1}\widehat\bS).$  
Because $\widetilde\bS$ has   the same LSD with 
$\mS_{n-1}$ defined in (\ref{Sn}) with $n$ substituted by $n-1$ ,  so that the matrix $\mS$  
 has the same LSD  as  $\mS_{n-1}$ due to the  positive definiteness  of $\bS_0$.  Therefore, it is natural to use $n-1$  instead of $n$ by \cite{ZhengBaiYao} when the  Lemma \ref{CLT} is applied  to  amending  the score test concerned the eigenvalues  of $\mS$. 
 Let 
\begin{equation}
\widetilde{\mb{RST}}(\chi, \bS_0)=\frac{n}{2}\mb{tr}[(\mS-\mb{\bI}_p)^2], \label{RSTt}
\end{equation}
then the correction to Rao's score test is hold in the following theorem:  

\begin{theorem}
Suppose that the conditions of Lemma \ref{CLT} hold, for hypothesis test $H_0 :  \bS  =\bS_0$,  $\widetilde{\mb{RST}}(\chi, \bS_0)$ is defined as (\ref{RSTt}),  set $p/(n-1)=q_n \rightarrow q \in [0, +\infty) $ ,  $q_n \neq 1$  and  $g(x)= (x-1)^2$.  Then,  under $H_0$ and when $n\rightarrow\infty$, the correction to Rao's score test statistics is 
\begin{equation}
  CRST(\chi, \bS_0)=\upsilon(g)^{-\frac{1}{2}}\left[ \frac{2}{n}\widetilde{\mb{RST}}(\chi, \bS_0)-p \cdot
    F^{q_n}(g)- \mu(g)\right] \Rightarrow N \left( 0,
  1\right),
  \label{CRST1}
\end{equation}
where $F^{q_n}$ is the Mar\v{c}enko-Pastur law of  index  $ q_n$, and $ F^{q_n}(g), \mu(g)$ and $\upsilon(g)$ are calculated in (\ref{limitRST}), (\ref{meanRST}) and (\ref{varRST}),  respectively.
\label{CRSTth}
\end{theorem}
\begin{proof}
By  the  derivation  (\ref{RSTt}), we have 
\begin{eqnarray*}
 \frac{2}{n}\widetilde{\mb{RST}}(\chi, \bS_0)&=&\mb{tr}[(\mS-\mb{\bI}_p)^2] \\
  &=& \sum\limits_{i=1}^{p} \left(\lambda_i^{\mS}-1\right)^2= p \cdot \int (x-1)^2
  \md F^{\mS}_n(x)\\
  &=&p \cdot \int g(x) \md\left(F^{\mS}_n(x)-F^{q_n}(x)\right) +p \cdot
  F^{q_n}(g),
\end{eqnarray*}
where  $(\lambda_i^{\mS}),  i=1,\cdots, p$ and $F_n^{\mS}$   are the eigenvalues and the   ESD of the matrix $\mS$, respectively.
$F^{q_n} (g)$ denotes the integral of the function $g(x) $ by the density corresponding to the Mar\v{c}enko-Pastur law of index $q_n$, that is 
\begin{eqnarray}
F^{q_n}(g)=\int_{-\infty}^{\infty}g(x) \mb{d} F^{q_n}(x) =q_n ,   \quad \text{if } q_n \neq 1,\label{limitRST}
\end{eqnarray}
which is calculated
in the Appendix \ref{A3}.
As the definition in (\ref{Gdef}), we have 
\begin{equation}
 G_n(g)= p \cdot \int g(x) d\left(F^{\mS}_n(x)-F^{q_n}(x)\right)=\frac{2}{n}\widetilde{\mb{RST}}(\chi, \bS_0)-p \cdot F^{q_n}(g).\label{singlESD-pLSD}
\end{equation}
By Lemma \ref{CLT}, $G_n(g)$  weakly converges  to a Gaussian vector  with the mean
\begin{equation}
\mu(g)=(\kappa-1)q+\beta q \label{meanRST}
\end{equation}
and variance
\begin{equation}
\upsilon(g)=2\kappa q^2(1+2q)+4\beta q^3,\label{varRST}
\end{equation}
which are calculated in  the Appendix \ref{A3}.
 Then, by (\ref{singlESD-pLSD}) we arrive at
\[
\frac{2}{n}\widetilde{\mb{RST}}(\chi, \bS_0)-p\cdot F^{q_n}(g)~\Rightarrow~ N\left(\mu(g),
\upsilon(g)\right),\label{1connec}
\]
 Finally, 
\begin{eqnarray}
 CRST(\chi, \bS_0)=\upsilon(g)^{-\frac{1}{2}}\left[ \frac{2}{n}\widetilde{\mb{RST}}(\chi, \bS_0)-p \cdot
    F^{q_n}(g)- \mu(g)\right] \Rightarrow N \left( 0,
  1\right)\nonumber
\end{eqnarray}
\end{proof}

For  the identity and  sphericity  hypothesis tests, we have  the following  corollaries:
\begin{corollary}
For testing $H_0 : \bS=\bI_p $ with no constrains on $\bmu$,  the conclusion of Theorem \ref{CRSTth} still holds,
only with the test  statistic $\widetilde{\mb{RST}}(\chi, \bI_p)$ in (\ref{CRST1}) is revised by
$$
\widetilde{\mb{RST}}(\chi, \bI_p)=\frac{n}{2}\mb{tr}[(\frac{n}{n-1}\widehat{\bS}-\mb{\bI}_p)^2] \label{CRST2}.
$$
\end{corollary}

\begin{corollary}
For testing $H_0 : \bS=\gamma\bI_p $ with no constrains on $\bmu$,   the conclusion of Theorem \ref{CRSTth} still holds,
only with the test statistic  $\widetilde{\mb{RST}}(\chi, \bI_p)$ in (\ref{CRST1}) is revised by

$$
\widetilde{\mb{RST}}(\chi, \gamma\bI_p)=\frac{n}{2}\mb{tr}[(\hat\gamma^{-1}\frac{n}{n-1}\widehat{\bS}-\mb{\bI}_p)^2]
$$
where  $\hat\gamma=\displaystyle\frac{\mb{tr}(\frac{n}{n-1}\widehat\bS)}{p}$ is the maximum likelihood estimator of $\gamma$.
\end{corollary}

%%%%%%%%%%%%%%%%%%%%%%%%%%%%%%%%%%%%%%%%%%%%%%%%%%%
%   Simulation Study
%%%%%%%%%%%%%%%%%%%%%%%%%%%%%%%%%%%%%%%%%%%%%%%%%%%
 
\section{Simulation Study}\label{Sim}

Simulations  are  conducted in this section to evaluate the correction to Rao's score test (CRST)   that we proposed.
To compare the performance, we also  present the corresponding simulation results of the test in \cite{Chen} (SCT),  the test in  \cite{Cai} (TCT)
and the  classical Rao's score test in \cite{Rao} (RST).  We consider the  identity hypothesis test  $H_0 : \bS=\bI_p $, and generate $\mb{i.i.d}$ random samples $\chi=(\bx_1,\cdots,\bx_n)$ from a   $p$-dimensional  
random vector $\bX$ following 
two scenarios of  the populations   under the null hypothesis,
\begin{itemize}
\item  Gaussian Assumption:   random vector $\bX$ follows a $p$-dimensional  normal  distribution with mean $\mu_0\bone_p$ and covariance matrix $\bI_p$, where $\mu_0=2$ and $\bone_p$ denotes a vector with that all elements  are 1. % It means $\kappa=2, \beta=0$ for the case of  real Gaussian variables here.
\item  Gamma Assumption:    random vector $\bX=(X_1, \cdots, X_p)'$    and the components are independent and  identically  distributed as  Gamma (4,0.5),  so that   each of the random variables $X_i$ also has mean 2 and    variance 1.   %For real Gamma variables here, we obtained that $\kappa=2, \beta=\frac{3}{2}$ for this  setting.
\end{itemize}

  For each set of the scenarios, we report both empirical Type I errors and powers with 10,000 replications at $\alpha=0.05$ significance  level. Different pair values of $p, n$ are selected at a wide rage regardless of  the functional expression  or   limiting behavior between them.  The mean parameter is supposed to be unknown and substituted by the sample mean during the calculations.

To calculate the empirical powers of the tests, two alternatives are designed %
 in the simulations. In the first alternative,  two different sample sets are  provided   for the corresponding scenarios.  For Gaussian assumption, the  samples are independently generated  from  the random vector $\bX$ following the normal distribution with mean  vector $\mu_0\bone_p$ and covariance matrix $\bS=\mb{diag} ( 2\cdot\bone_{[v_0p]},\bone_{p-[v_0p]})$, where $\mu_0=2$,  $v_0=0.02$ are varying constants and $[ \cdot  ]$ denotes the integer truncation function. 
 For Gamma assumption,  the  samples are still randomly selected   from  the random vector  $\bX=(X_1, \cdots, X_p)'$   with independent  components. Each component of the  front  part $(X_1,\cdots, X_{[ v_0p ]})$ is distributed as Gamma(2,1),  whereas each of the components in the rest part  $(X_{[ v_0p ]+1},\cdots, X_{[ p-v_0p ]})$ follows  Gamma(4,0.5), where $v_0=0.04$. In the second alternative,    the  samples for Gaussian assumption are independently drawn  from  the normal distribution with mean  vector $\mu_0\bone_p$ and covariance matrix $\bS=\mb{diag} ( (1+20/\sqrt{np})\cdot\bone_{[v_0p]},\bone_{p-[v_0p]})$, where $\mu_0=2$,  $v_0=0.25$.   The  samples for Gamma assumption are followed   the distribution of  the random vector  $\bX=(X_1, \cdots, X_p)'$,  which satisfied that  the  components  are independent and each component in the  front  part $(X_1,\cdots, X_{[ v_0p ]})$ is distributed as Gamma($\frac{4}{1+20/\sqrt{np}},\frac{1+20/\sqrt{np}}{2}$),  whereas each  component of the rest part  $(X_{[ v_0p ]+1},\cdots, X_{[ p-v_0p ]})$ follows  Gamma(4,0.5), where $v_0=0.25$.
 
 Simulation results of empirical Type I errors and powers for the first alternative  are listed in the Table  \ref{tab:1}, and the empirical  powers for the second alternative is represented in
 Table \ref{table2}.

 \begin{table}[htbp]
  \centering\caption{ Empirical sizes and powers(in brackets)  of the comparative tests for  $H_0 : \bS=\bI_p$  at$~\alpha=0.05$ significance level  for normal and gamma random vectors with 10,000 \\    replications.  The alternative hypothesis is    $\bS=\mb{diag} ( 2\cdot\bone_{[v_0p]},\bone_{p-[v_0p]})$ with \\$v_0 =$ $ 0.02$ for Normal variables and $v_0=0.04$ for Gamma variables.\label{tab:1} }
  \begin{tabularx}{12cm}{rXXXX}   
\toprule
 &  CRST                     &SCT                      &TCT                 &RST                              \\[-0.75 ex]
$p$&proposed &\multicolumn{2}{c}{$n$=19}&\\[-0.75 ex]
%\cline{4-5} 
    \midrule
    &\multicolumn{4}{c}{\text{Normal random vectors:  Type I error~~(Power)}}\\[-0.75 ex]
17     &0.0732~~(0.2785)           &0.0748~~(0.2034)~~  &0.1057~~(0.2497)              &0.0973~~(0.3285)~~   \\[-0.75 ex]
20     &0.0715~~(0.2511)           &0.0765~~(0.1821)~~  &0.1216~~(0.2494)               &0.1018~~(0.3120)~~   \\[-0.75 ex]
40     &0.0628~~(0.1638)           &0.0752~~(0.1246)~~  &0.2540~~(0.3425)               &0.1129~~(0.2487)~~   \\[-0.75 ex]
80     &0.0586~~(0.1148)           &0.0720~~(0.0951)~~  &0.6664~~(0.7115)               &0.1070~~(0.1893)~~   \\[-0.75 ex]
160   &0.0566~~(0.1955)           &0.0718~~(0.1078)~~  &0.9985~~(0.9990)               &0.0809~~(0.2523)~~    \\[-0.75 ex]
320   &0.0573~~(0.2791)           &0.0726~~(0.1046)~~  &1.0000~~(1.0000)               &0.0419~~(0.2434)~~   \\[-0.75 ex]
    &\multicolumn{4}{c}{\text{Gamma random vectors:  Type I error ~~(Power) }}\\[-0.75 ex]
17     &0.0973~~(0.0986)            &0.0942~~(0.0895)~~  &0.1293~~(0.1324)              &0.2189~~(0.2173)~~      \\[-0.75 ex]
20     &0.0930~~(0.0943)            &0.0913~~(0.0887)~~  &0.1450~~(0.1465)               &0.2234~~(0.2255)~~   \\[-0.75 ex]
40     &0.0774~~(0.1611)            &0.0788~~(0.1316)~~  &0.2794~~(0.3584)               &0.2282~~(0.3632)~~   \\[-0.75 ex]
80     &0.0651~~(0.2532)            &0.0745~~(0.1578)~~  &0.6858~~(0.7897)               &0.2098~~(0.4960)~~   \\[-0.75 ex]
160   &0.0585~~(0.3275)            &0.0734~~(0.1552)~~  &0.9988~~(0.9995)               &0.1667~~(0.5511)~~   \\[-0.75 ex]
320   &0.0546~~(0.4719)            &0.0714~~(0.1497)~~  &1.0000~~(1.0000)               &0.1084~~(0.6082)~~   \\[-0.75 ex]
  \bottomrule
  \end{tabularx}
  
   \begin{tabularx}{12cm}{rXXXX}
%\toprule
$p$&\multicolumn{4}{c}{$n$=39}\\[-0.75 ex]
    \midrule
    &\multicolumn{4}{c}{\text{Normal random vectors:  Type I error ~~(Power) }}\\[-0.75 ex]
20     &0.0687~~(0.3986)~~         &0.0701~~(0.2904)~~  &0.0828~~(0.3168)~~  &0.0892~~(0.4451)~~   \\[-0.75 ex]
37     &0.0633~~(0.2558)~~         &0.0638~~(0.1771)~~  &0.1179~~(0.2629)~~  &0.1043~~(0.3451)~~   \\[-0.75 ex]
40     &0.0614~~(0.2448)~~         &0.0637~~(0.1680)~~  &0.1246~~(0.2644)~~  &0.1046~~(0.3370)~~  \\[-0.75 ex]
80     &0.0573~~(0.1529)~~         &0.0622~~(0.1061)~~  &0.2591~~(0.3465)~~  &0.1151~~(0.2581)~~  \\[-0.75 ex]
160   &0.0549~~(0.3022)             &0.0603~~(0.1279)~~  &0.6532~~(0.7761)~~  &0.1091~~(0.4420)~~  \\[-0.75 ex]
320   &0.0530~~(0.4508)             &0.0614~~(0.1246)~~   &0.9955~~(0.9985)    &0.0832~~(0.5412)~~  \\[-0.75 ex]
    &\multicolumn{4}{c}{\text{Gamma random vectors:  Type I error ~~(Power) }}\\[-0.75 ex]
20     &0.0937~~(0.0945)~~            &0.0891~~(0.0901)~~  &0.1013~~(0.1003)~~  &0.2253~~(0.2316)~~  \\[-0.75 ex]
37     &0.0789~~(0.2410)~~            &0.0741~~(0.1902)~~  &0.1339~~(0.2830)~~  &0.2398~~(0.4734)~~  \\[-0.75 ex]
40     &0.0761~~(0.2246)~~            &0.0729~~(0.1707)~~  &0.1415~~(0.2656)~~  &0.2414~~(0.4649)~~  \\[-0.75 ex]
80     &0.0668~~(0.3958)~~            &0.0664~~(0.2297)~~  &0.2748~~(0.5244)~~  &0.2441~~(0.6826)~~  \\[-0.75 ex]
160   &0.0581~~(0.5241)                &0.0629~~(0.2250)~~  &0.6615~~(0.8553)~~  &0.2189~~(0.7913)~~  \\[-0.75 ex]
320   &0.0540~~(0.7301)                &0.0611~~(0.2220)~~   &0.9954~~(0.9996)    &0.1764~~(0.8957)~~  \\[-0.75 ex]
  \bottomrule
  \end{tabularx}
\end{table}

{
\renewcommand{\thetable}{}
\renewcommand\tablename{Table~1 }  
\begin{table}[htbp]
  \centerline{T{\tiny ABLE}~1~~(cont.)}
  \begin{tabularx}{12cm}{rXXXX}   
\toprule
  &   CRST                                        &SCT                      &TCT                 &RST                        \\[-0.75 ex]
  $p$&proposed &\multicolumn{2}{c}{$n$=79}&\\[-0.75 ex]
%\cline{4-5} 
    \midrule
    &\multicolumn{4}{c}{\text{Normal random vectors:  Type I error ~~(Power)} }\\[-0.75 ex]
20      &0.0640~~(0.6563)~~          &0.0645~~(0.5322)~~  &0.0690~~(0.5353)~~  &0.0759~~(0.6743)~~   \\[-0.75 ex]
40      &0.0591~~(0.4096)~~          &0.0594~~(0.2874)~~  &0.0832~~(0.3362)~~  &0.0885~~(0.4841)~~   \\[-0.75 ex]
77      &0.0572~~(0.2401)~~          &0.0568~~(0.1556)~~  &0.1226~~(0.2659)~~  &0.1064~~(0.3529)~~   \\[-0.75 ex]
80      &0.0568~~(0.2328)~~          &0.0577~~(0.1487)~~  &0.1258~~(0.2620)~~  &0.1072~~(0.3458)~~   \\[-0.75 ex]
160    &0.0540~~(0.5045)~~          &0.0554~~(0.2068)~~  &0.2573~~(0.5345)~~  &0.1164~~(0.6625)~~   \\[-0.75 ex]
320    &0.0511~~(0.7117)~~          &0.0564~~(0.2036)~~  &0.6458~~(0.8616)~~  &0.1097~~(0.8301)~~   \\[-0.75 ex]
    &\multicolumn{4}{c}{\text{Gamma random vectors:  Type I error ~~(Power) }}\\[-0.75 ex]
20      &0.0898~~(0.0919)~~          &0.0832~~(0.0816)~~  &0.0861~~(0.0839)~~  &0.2175~~(0.2156)~~   \\[-0.75 ex]
40      &0.0748~~(0.3623)~~          &0.0698~~(0.2946)~~  &0.0959~~(0.3415)~~  &0.2384~~(0.6094)~~   \\[-0.75 ex]
77      &0.0642~~(0.6529)~~          &0.0618~~(0.4423)~~  &0.1303~~(0.5758)~~  &0.2487~~(0.8698)~~  \\[-0.75 ex]
80      &0.0633~~(0.6473)~~          &0.0615~~(0.4292)~~  &0.1337~~(0.5821)~~  &0.2494~~(0.8656)~~   \\[-0.75 ex]
160    &0.0579~~(0.8028)~~          &0.0579~~(0.4377)~~  &0.2668~~(0.7503)~~  &0.2487~~(0.9543)~~   \\[-0.75 ex]
320    &0.0542~~(0.9493)~~          &0.0551~~(0.4533)~~  &0.6524~~(0.9525)~~  &0.2231~~(0.9924)~~   \\[-0.75 ex]
  \bottomrule
  \end{tabularx}
  
   \begin{tabularx}{12cm}{rXXXX}   
%\toprule
$p$& \multicolumn{4}{c}{$n$=159}\\[-0.75 ex]
    \midrule
    &\multicolumn{4}{c}{\text{Normal random vectors: Type I error~~(Power)}}\\[-0.75 ex]
20     &0.0685~~(0.9178)~~           &0.0682~~(0.8598)~~  &0.0693~~(0.8657)~~  &0.0713~~(0.9178)~~   \\[-0.75 ex]
40     &0.0526~~(0.7159)~~           &0.0531~~(0.5891)~~  &0.0616~~(0.6092)~~  &0.0629~~(0.7563)~~  \\[-0.75 ex]
80     &0.0539~~(0.4325)~~           &0.0586~~(0.2836)~~  &0.0839~~(0.3610)~~  &0.0874~~(0.5238)~~   \\[-0.75 ex]
157   &0.0558~~(0.7986)~~           &0.0563~~(0.4415)~~  &0.1198~~(0.5954)~~  &0.1032~~(0.8869)~~   \\[-0.75 ex]
160   &0.0565~~(0.8051)~~           &0.0541~~(0.4443)~~  &0.1260~~(0.5989)~~  &0.1098~~(0.8902)~~   \\[-0.75 ex]
320   &0.0505~~(0.9530)~~           &0.0522~~(0.4269)~~  &0.2545~~(0.7714)~~  &0.1109~~(0.9736)~~   \\[-0.75 ex]
    &\multicolumn{4}{c}{\text{Gamma random vectors: Type I error~~(Power)}}\\[-0.75 ex]
20     &0.0762~~(0.0889)~~          &0.0760~~(0.0745)~~  &0.0766~~(0.0747)~~  &0.1983~~(0.2212)~~   \\[-0.75 ex]
40     &0.0647~~(0.6090)~~          &0.0595~~(0.5397)~~  &0.0728~~(0.5521)~~  &0.2201~~(0.7831)~~   \\[-0.75 ex]
80     &0.0586~~(0.9145)~~          &0.0682~~(0.7869)~~  &0.0997~~(0.8267)~~  &0.2493~~(0.9821)~~   \\[-0.75 ex]
157   &0.0575~~(0.9783)~~          &0.0594~~(0.8230)~~  &0.1175~~(0.8963)~~  &0.2502~~(0.9990)~~   \\[-0.75 ex]
160   &0.0565~~(0.9785)~~          &0.0576~~(0.8243)~~  &0.1362~~(0.9124)~~  &0.2475~~(0.9997)~~   \\[-0.75 ex]
320   &0.0536~~(0.9998)~~          &0.0528~~(0.8566)~~  &0.2523~~(0.9716)~~  &0.2673~~(1.0000)~~   \\[-0.75 ex]
  \bottomrule
  \end{tabularx}
  \end{table}
} 

 \begin{table}[htbp]
  \centering\caption{ Empirical powers of the comparative tests for  $H_0 : \bS=\bI_p$  at$~\alpha=0.05$ significance level  for normal and gamma random vectors with 10,000    replications.  The alternative hypothesis is    $\bS=\mb{diag} ( (1+20/\sqrt{np})\cdot\bone_{[v_0p]},\bone_{p-[v_0p]})$ with $v_0 = 0.25$. \label{table2} }
  \begin{tabularx}{12cm}{rXXXXrXXXX}   
\toprule
 &    CRST                                     &SCT                      &TCT                 &RST          &     &      CRST                                       &SCT                      &TCT                 &RST                  \\[-0.75 ex]
 $p$ &proposed &\multicolumn{3}{c}{  ~}&$p$  &proposed &\multicolumn{3}{c}{~}\\[-0.75 ex]
     \midrule
       &\multicolumn{9}{c}{\text{Normal random vectors}}\\[-0.75 ex]
  &\multicolumn{4}{c}{$n$=19}&   &\multicolumn{4}{c}{$n$=39}\\[-0.75 ex]
%\cline{4-5}  
17     &0.9095     &0.6115     &0.6651    &0.9288        &20     &0.9395     &0.6372     &0.6541    &0.9501   \\[-0.75 ex]
20     &0.9250     &0.5862     &0.6648    &0.9440        &37     &0.9220     &0.3900     &0.4960    &0.9529    \\[-0.75 ex]
40     &0.9322     &0.3787     &0.6615    &0.9646        &40     &0.9389     &0.3768     &0.4949    &0.9639    \\[-0.75 ex]
80     &0.9413     &0.2446     &0.8689    &0.9684        &80     &0.9448     &0.2094     &0.5059    &0.9743    \\[-0.75 ex]
160   &0.9508     &0.1608     &0.9996    &0.9674        &160   &0.9526     &0.1327     &0.7916    &0.9776    \\[-0.75 ex]
320   &0.9545     &0.1159     &1.0000    &0.9421        &320   &0.9591     &0.0982     &0.9972    &0.9755    \\[-0.75 ex]
   &\multicolumn{4}{c}{$n$=79}&   &\multicolumn{4}{c}{$n$=159}\\[-0.75 ex]
20     &0.9356     &0.6734     &0.6645    &0.9389        &20       &0.9260     &0.6881     &0.6734   &0.9231    \\[-0.75 ex]
40     &0.9202     &0.3757     &0.4235    &0.9428        &40       &0.8794     &0.3689     &0.3745    &0.9102    \\[-0.75 ex]
77     &0.9206     &0.1954     &0.3115    &0.9579        &80       &0.8912     &0.1807     &0.2312    &0.9319    \\[-0.75 ex]
80     &0.9225     &0.1952     &0.3167    &0.9610        &157     &0.9081     &0.0876     &0.1976    &0.9490    \\[-0.75 ex]
160   &0.9426     &0.1250     &0.3991    &0.9756        &160     &0.9104     &0.1104     &0.2125    &0.9558    \\[-0.75 ex]
320   &0.9551     &0.0893     &0.7209    &0.9823        &320     &0.9348     &0.0767     &0.3372    &0.9783    \\[-0.75 ex]
  \bottomrule
  \end{tabularx}
  
   \begin{tabularx}{12cm}{rXXXXrXXXX}
%\toprule
 %&  proposed   CRST                                       &SCT                      &TCT                 &RST          &     &  proposed   CRST                                       &SCT                      &TCT                 &RST                  \\

   % \midrule
    &\multicolumn{9}{c}{\text{Gamma random vectors }}\\[-0.75 ex]
      &\multicolumn{4}{c}{$n$=19}&  &\multicolumn{4}{c}{$n$=39}\\[-0.75 ex]
17     &0.7811     &0.5381     &0.6074    &0.8819        &20     &0.8491    &0.5829     &0.6067    &0.9352    \\[-0.75 ex]
20     &0.8000     &0.5184     &0.6220    &0.9039        &37     &0.8117     &0.3854     &0.4921    &0.9363    \\[-0.75 ex]
40     &0.8104     &0.3781     &0.6746    &0.9330        &40     &0.8311     &0.3777     &0.4969    &0.9469    \\[-0.75 ex]
80     &0.8108     &0.2343     &0.8768    &0.9388        &80     &0.8153     &0.2146     &0.5179    &0.9534    \\[-0.75 ex]
160   &0.8191     &0.1622     &0.9998    &0.9319        &160     &0.8225     &0.1345     &0.7969    &0.9577    \\[-0.75 ex]
320   &0.8159     &0.1232     &1.0000    &0.8961        &320     &0.8270     &0.1001     &0.9977    &0.9457    \\[-0.75 ex]
   &\multicolumn{4}{c}{$n$=79}&   &\multicolumn{4}{c}{$n$=159}\\[-0.75 ex]
20     &0.8652     &0.6203     &0.6196    &0.9386        &20     &0.8769     &0.6548      &0.6437    &0.9543    \\[-0.75 ex]
40     &0.8275     &0.3694     &0.4176    &0.9430        &40     &0.8104     &0.3545      &0.3653    &0.9434    \\[-0.75 ex]
77     &0.8029     &0.1981     &0.3250    &0.9529        &80     &0.8096     &0.2211       &0.2765    &0.9650    \\[-0.75 ex]
80     &0.8139     &0.1959     &0.3243    &0.9562        &157     &0.7776     &0.1109     &0.2107    &0.9610    \\[-0.75 ex]
160   &0.8170     &0.1259     &0.4034    &0.9638        &160     &0.8102     &0.1312     &0.2322    &0.9546    \\[-0.75 ex]
320   &0.8256     &0.0896     &0.7243    &0.9630        &320     &0.8296     &0.1095     &0.3540    &0.9837    \\ [-0.75 ex]
 \bottomrule
  \end{tabularx}
\end{table}

As seen from the Table \ref{tab:1},  the empirical Type I errors of our proposed test CRST for both scenarios
 are almost around the nominal size 5\%,  and it converges   to  the  nominal level rapidly as the dimension $p$
 approaches   infinity,  even for small $n$.  Although,  the empirical sizes of
  the proposed CRST is slightly higher  for the case of $p=17$ or 20 under the Gamma assumption, it  can be accepted with comparison to the other tests and be understood due to both asymptotic and nonparametric.

For  a further comparison, it is limited to several  aspects. 
 First, the Rao's score test and our proposed test both give  a good performance when  $p$  is very small   under the normal   assumption.   However, the empirical  sizes of the Rao's score test deviate from the nominal level as $p$ increases, and it shows a even worse result  under the Gamma distribution assumption,  where the proposed CRST is  still active. 
Another interesting note is that the Rao's score test  has a resilient power for the normal cases when $p$ is much higher than $n$, for example ($p=160, n=19$).

Second, for  small and moderate dimensions like $p=20$  or
40 with higher sample size  $n=79$ or 159,  the empirical Type I errors
 of  the TCT in  \cite{Cai} behave well. 
However, the TCT  leads to a dramatically high empirical size as the dimension $p$ increases much higher,   especially for "large $p$ , small $n$" such as (p=160,n=19), though it has the optimal powers.  Meanwhile, the proposed  CRST remains accurate.
 
Last, compared to the SCT in  \cite{Chen}, our proposed CRST have more closer empirical sizes to 5\% with growing dimension  $p$, especially  for the small sample sizes. Furthermore, the powers of  proposed CRST uniformly dominates that of the SCT 
over the entire range.  For example, the powers of  proposed CRST rise rapidly up to 1 as $p$ increases in the case of $n=79$ under the Gamma assumption, while those of  the SCT remains less than 0.5 even if the sample size is not quite small.

Finally, It must be pointed out that the proposed CRST cannot be use for the case $q_n=1$, but it remains in force even if  $q=1$, which means the $q_n$ could be very close to 1 by two sides. So we choose a different $p$ for each $n$, which makes $q_n \rightarrow 1^{-}$, for example ($p$=17,$n$=19) or  ($p$=77,$n$=79). Also, the cases as  ($p$=20,$n$=19) or  ($p$=80,$n$=79) are chosen for $q_n \rightarrow 1^{+}$. As seen from the results, the  proposed CRST performs well even if  
 $q_n \rightarrow 1$.
 
 Table \ref{table2} shows a more apparent comparison advantage under the second alternative.  The higher empirical powers  of RST  and SCT don't  make sense  because their empirical sizes are much higher. Moreover,
 the powers of the TCT  decline sharply, even near to 0.1, when the dimension $p$ rises up. Whereas, the proposed CRST gives the powers around 0.9 at the eligible empirical sizes.  
 
For a more intuitive understanding, take the cases ($n$=39,$p$=80) and ($n$=39,$p$=320) as  an example,  Figure \ref{fig:1} portrays a  dynamic view of  the powers for the first alternative  under the  Gamma assumption by the  varying  parameter $v_0$   from 0 to 0.10.  
Figure \ref{fig:2}  describes  the powers for the second alternative under the Gamma assumption by the  varying  parameter $v_0$   from 0 to 0.5.  Because $v_0$  depicts  the distance between the null and alternative hypothesis, so the starting point at $v_0=0$ is for the empirical sizes.  As shown in the picture, the proposed CRST  is a more sensitive  and powerful test with the accurate  empirical sizes.

\begin{figure}[htbp]
\begin{center}
\includegraphics[width = .45\textwidth]{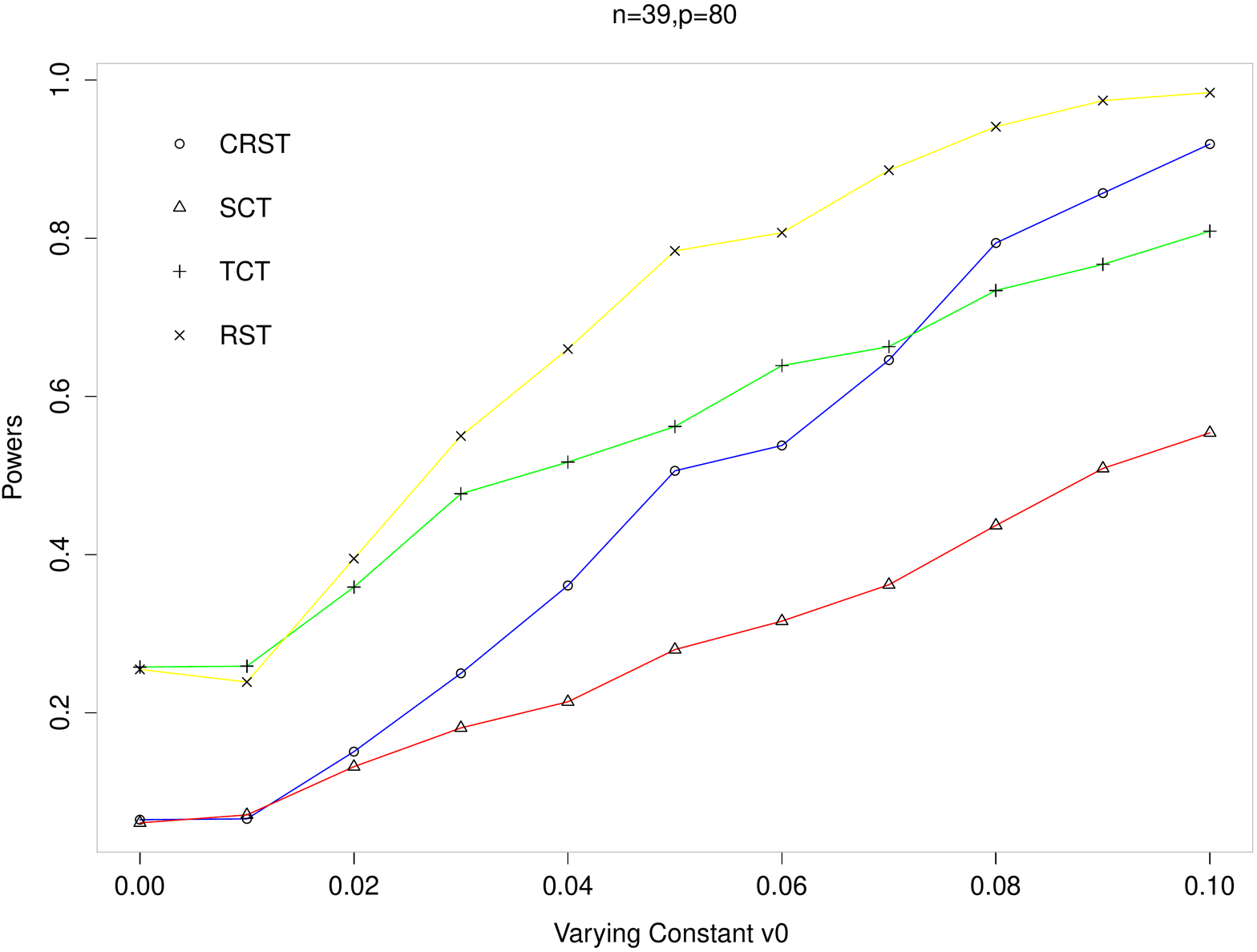}\quad \includegraphics[width = .45\textwidth] {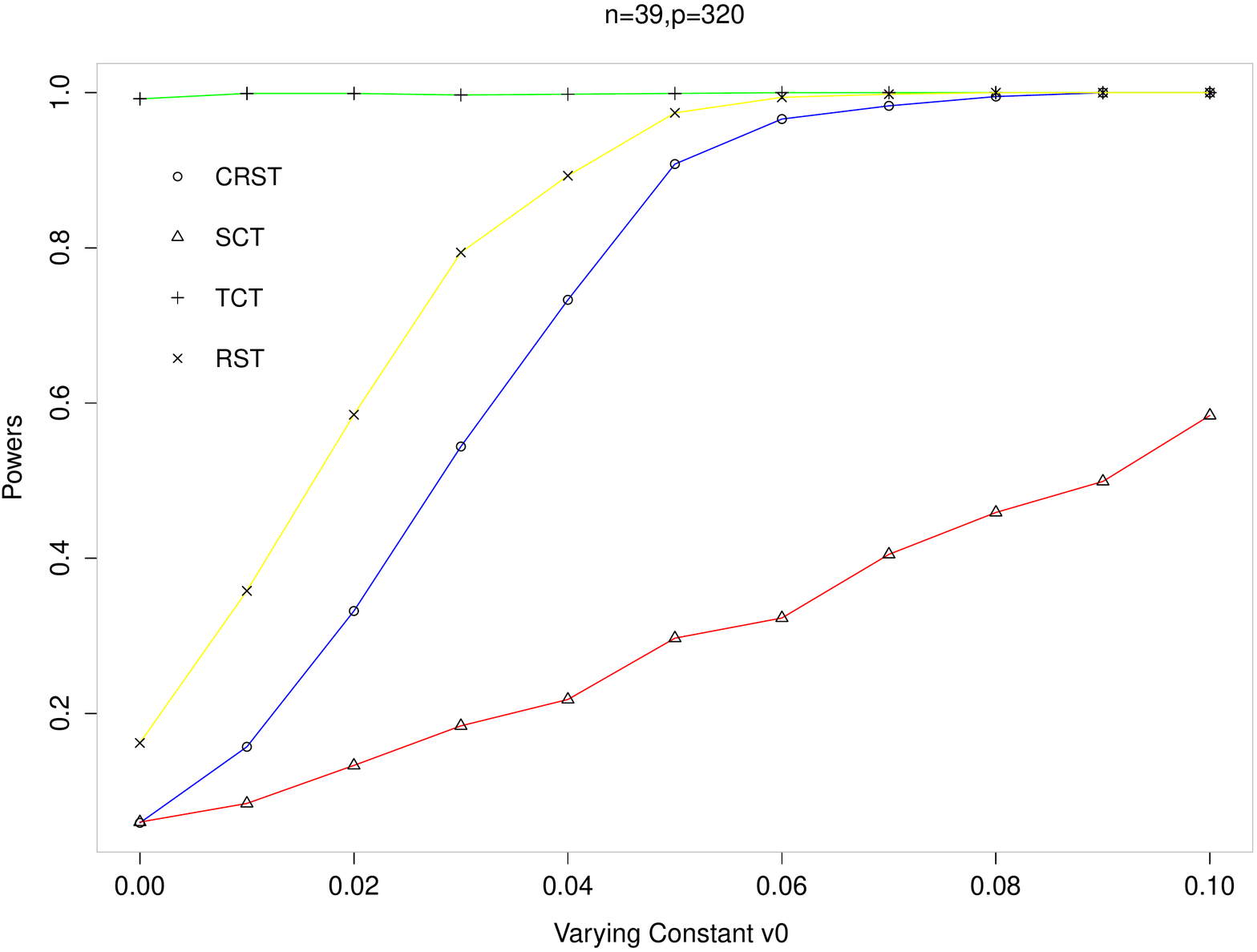}
\caption{Empirical sizes  and powers  of the comparative tests for $H_0 : \bS=\bI_p $  at  $\alpha=0.05$  significance level based on 10,000 independent replications of Gamma Assumption.  The null and alternative hypothesis are   $\bS=\mb{diag} (2\cdot\bone_{[v_0p]},\bone_{p-[v_0p]})$ with $v_0$ varied from 0 to 0.10. Left: $n=39,p=80$; Right: $n=39,p=320$.}
\label{fig:1}
\end{center} 
\end{figure}

\begin{figure}[htbp]
\begin{center}
\includegraphics[width = .45\textwidth]{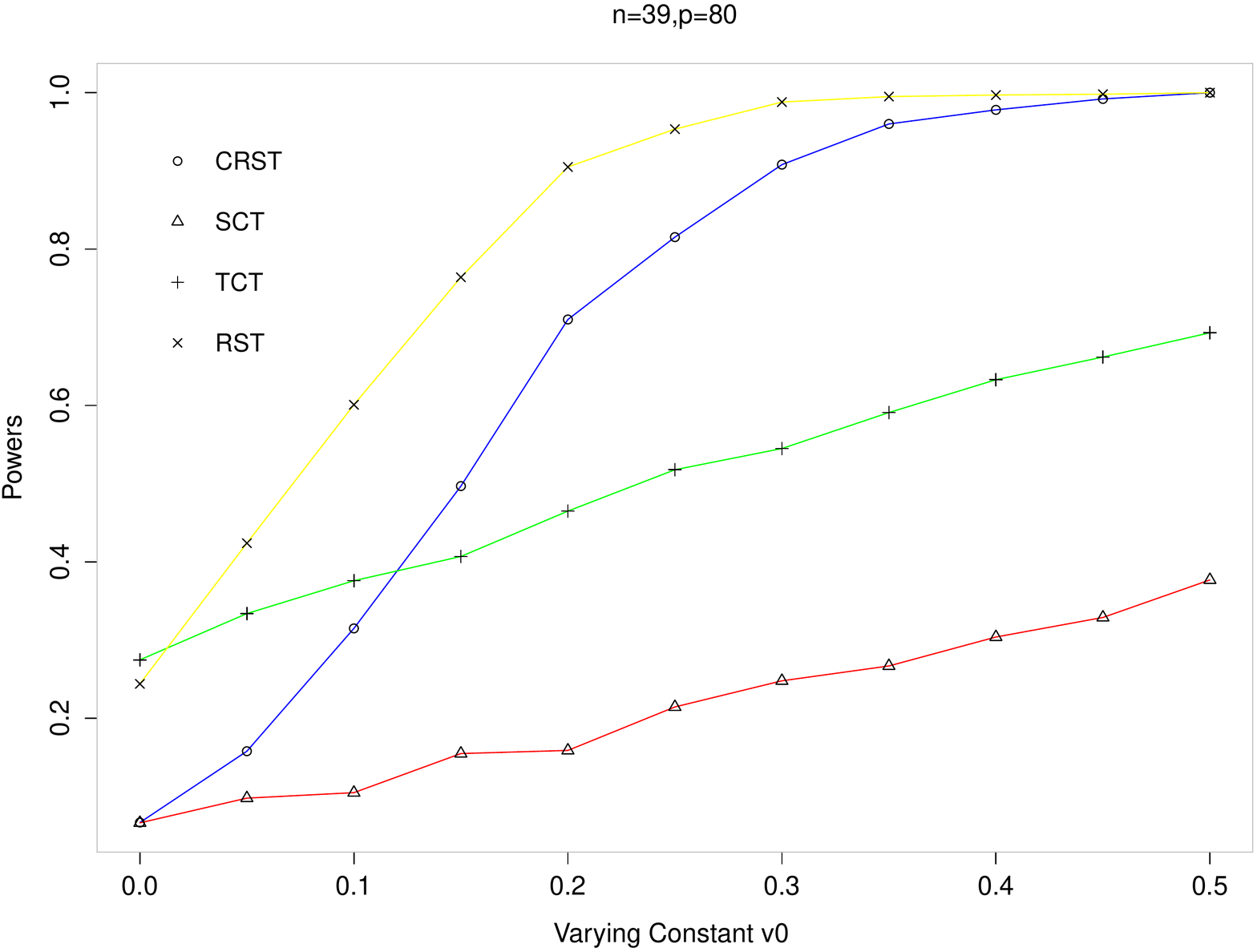}\quad \includegraphics[width = .45\textwidth] {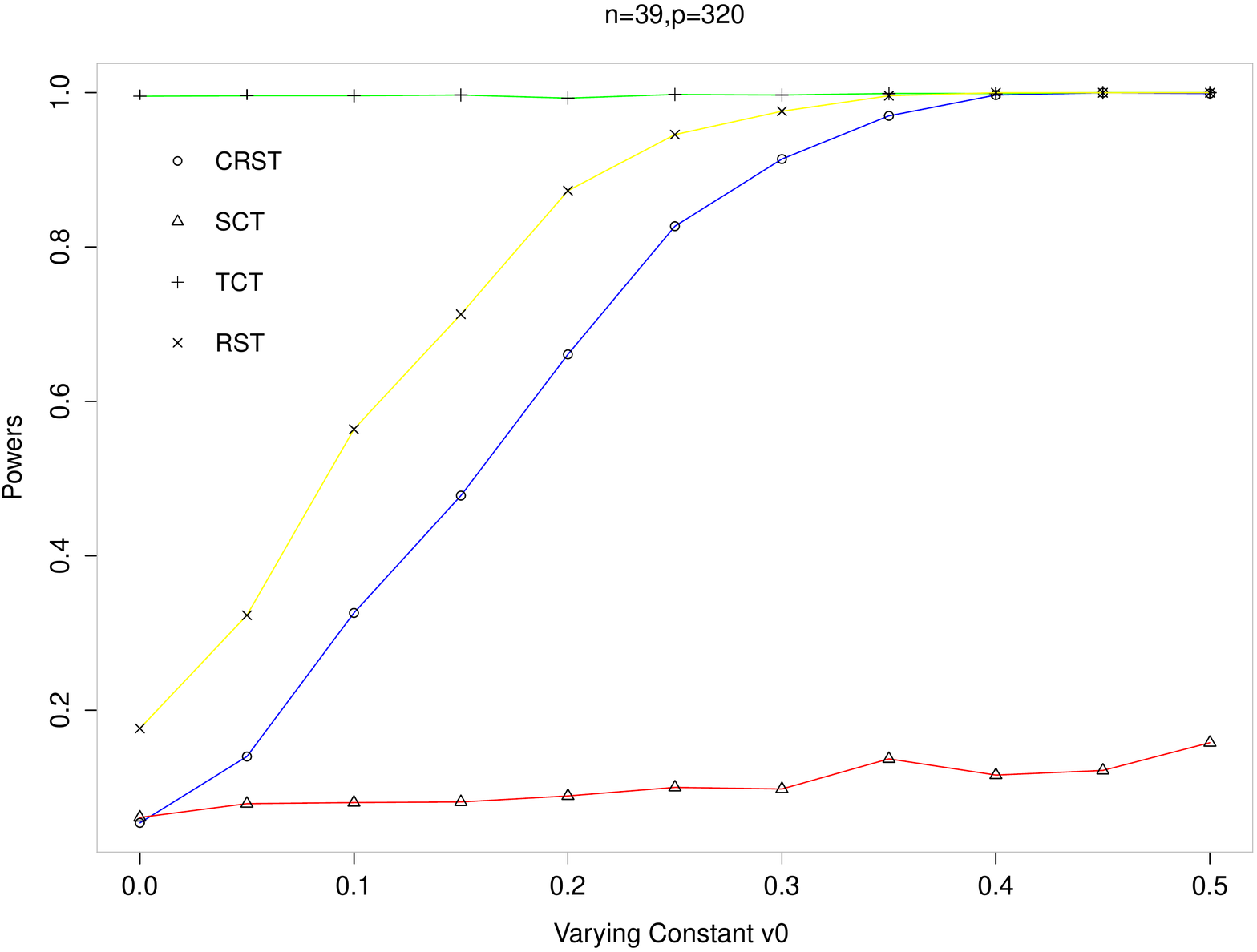}
\caption{Empirical sizes  and powers  of the comparative tests for $H_0 : \bS=\bI_p $  at  $\alpha=0.05$  significance level based on 10,000 independent replications of Gamma Assumption.  The null and alternative hypothesis are   $\bS=\mb{diag} ((1+20/\sqrt{np})\cdot\bone_{[v_0p]},\bone_{p-[v_0p]})$ with $v_0$ varied from 0 to 0.50. Left: $n=39,p=80$; Right: $n=39,p=320$.}
\label{fig:2}
\end{center} 
\end{figure}

%%%%%%%%%%%%%%%%%%%%%%%%%%%%%%%%%%%%%%%%%%%%%%%%%%%
%  Conclusion
%%%%%%%%%%%%%%%%%%%%%%%%%%%%%%%%%%%%%%%%%%%%%%%%%%%

\section{Conclusion}\label{Con}

In this paper, we propose a  new testing statistic for the large dimensional covariance structure test  based on amending Rao's score tests by RMT. Through  generalizing  the  CLT for LSS of  a large dimensional  sample covariance matrix in \cite{BS04}, we guarantee  the test proposed is feasible for the non-Gaussian variables in a wider range. Furthermore, the correction  to  Rao's score test can be also used in the case of ultra high dimensionality regardless of  the functional relationship between $p$ and $n$. 
It breaks the inherent thinking  that the corrections  by RMT are usually practicable when $p<n$,
and shows that it is  the corrected statistics we chose to decide whether the corrections by RMT can be used  in the case of $p>n$  rather than the tools we used in RMT.  So we believe that
 large dimensional  spectral analysis in RMT  will have  more application fields in light of different situations.

%%%%%%%%%%%%%%%%%%%%%%%%%%%%%%%%%%%%%%%%%%%%%%%%%%%
%  appendix
%%%%%%%%%%%%%%%%%%%%%%%%%%%%%%%%%%%%%%%%%%%%%%%%%%%

\appendix

\section{Derivations and Proofs.}\label{app}
\subsection{ Proofs of the derivation in (\ref{scorevec}) and (\ref{H22d}) \label{A1}} 

The  logarithm  of the density of the sample $\chi$ is written  as  
\begin{eqnarray*}
l(\chi,\btea)&=&-\frac{np}{2}\mb{ln}(2\pi)-\frac{n}{2}\ln|\bS|-\frac{1}{2}\sum\limits_{i=1}^{n}(\bx_i-\bmu)'\bS^{-1}(\bx_i-\bmu)\\
&=&-\frac{np}{2}\mb{ln}(2\pi)-\frac{n}{2}\ln|\bS|-\frac{1}{2}\sum\limits_{i=1}^{n}\mb{tr}\left(\bS^{-1}(\bx_i-\bmu)(\bx_i-\bmu)'\right),
\end{eqnarray*}
and  $\btea$ is denoted as $(\bmu', \mb{vec}(\bS)')' . $

For the first part of  (\ref{scorevec}),
by  the  formula $\displaystyle\frac{\md X'BX}{\md X}=(B+B')X$, where the $X$ is a vector and $B$ is a  matrix dependent on $X$,
 we have 
\begin{eqnarray*}
\frac{\md l(\chi,\btea)}{\md \bmu}&=&-\frac{1}{2}\sum\limits_{i=1}^{n} \frac{\md (\bx_i-\bmu)'\bS^{-1}(\bx_i-\bmu) }{\md \bmu}\\
&=&-\frac{1}{2}\sum\limits_{i=1}^{n} -2 \bS^{-1}(\bx_i-\bmu)\\
&=&n\bS^{-1}(\hat\bmu-\bmu)
\end{eqnarray*}
where $\hat\bmu=\displaystyle\frac{1}{n}\sum\limits_{i=1}^n\bx_i$.

For the second part of  (\ref{scorevec}),  by the following formulas 
\begin{eqnarray*}
&\frac{\md  \ln |X|}{\md X}=\mb{vec}((X^{-1})') \quad &    \frac{\md  X^{-1}}{\md X}=-((X^{-1})'\otimes X^{-1}) \\
&\frac{\md  \mb{tr} (B'X)}{\md X}=\mb{vec}(B) \quad &     (B' \otimes C) \mb{vec}(D)=\mb{vec}(CDB) 
\end{eqnarray*}
where $X,  B, C,D$ are all matrices. Then
we have 
\begin{eqnarray*}
\frac{\md l(\chi,\btea)}{\md \bS}&=&-\frac{n}{2}\frac{\md  \ln |\bS|}{\md \bS}-\frac{1}{2} \frac{\md \bS^{-1}}{\md \bS}\cdot\frac{\md  \mb{tr}\left(\bS^{-1}\sum\limits_{i=1}^{n} (\bx_i-\bmu)(\bx_i-\bmu)'\right)}{\md \bS^{-1}} \\
&=&-\frac{n}{2}\mb{vec}(\bS^{-1})+\frac{n}{2}(\bS^{-1}\otimes\bS^{-1})\mb{vec}(\bA) \\
&=&\frac{n}{2}\mb{vec}(\bS^{-1}(\bA\bS^{-1}-\bI_p))
\end{eqnarray*}
where $\bA=\displaystyle\frac{1}{n}\sum\limits_{i=1}^{n}(\bx_i-\bmu)(\bx_i-\bmu)'$. 
Thus 
\[
\frac{\md l(\chi,\btea)}{\md \mb{vec}(\bS)}= \mb{vec}\left(\frac{\md l(\chi,\btea)}{\md \bS}\right)=\frac{n}{2}\mb{vec}(\bS^{-1}(\bA\bS^{-1}-\bI_p)).\]
Therefore,  the score vector  for the sample is 
\[
U(\chi, \btea)=:\left(
\begin{array}{c}
U_1 (\chi, \btea) \\
U_2(\chi, \btea)
\end{array}
\right)=\left(
\begin{array}{c}
n\bS^{-1}(\hat\bmu-\bmu)  \\
\displaystyle\frac{n}{2}\mb{vec}(\bS^{-1}(\bA\bS^{-1}-\bI_p))
\end{array}
\right)
\]

Next consider the derivation  of  (\ref{H22d}).
By the definitions of the Hessian matrix  and score vector, we have 
\[H(\chi, \btea)=\frac{\md^2}{\md \btea^2}l(\chi,\btea)=\frac{\md U(\chi,\btea)}{\md \btea'}=: 
\left(
\begin{array}{ccc}
 H_{11} & H_{12}    \\
 H_{21} & H_{22}   
\end{array}
\right),
\] where the part of the parameter $\bS$ is
\begin{eqnarray*}
H_{22}&=&\frac{n}{2}\frac{\md \mb{vec}(\bS^{-1}(\bA\bS^{-1}-\bI_p))}{\md \mb{vec}'(\bS)}\\ \nonumber 
&=&\frac{n}{2}\frac{\md \mb{vec}(\bS^{-1})}{\md \mb{vec}'(\bS)} \frac{\md \mb{vec}(\bS^{-1}\bA\bS^{-1}-\bS^{-1})}{\md  \mb{vec}'(\bS^{-1})}\\\nonumber
&=& -\frac{n}{2} (\bS^{-1} \otimes \bS^{-1} )(\bA\bS^{-1} \otimes \bI_p+\bI_p\otimes \bA\bS^{-1}-\bI_{p^2}).
\end{eqnarray*}
Since 
\begin{eqnarray*}
&&\frac{\md (\mb{vec}(\bS^{-1}\bA\bS^{-1})-\mb{vec}(\bS^{-1}))}{\md  \mb{vec}'(\bS^{-1})}
=\frac{\md ( \bA\bS^{-1}\otimes \bI_p)\mb{vec}(\bS^{-1})}{\md  \mb{vec}'(\bS^{-1})}- \frac{\md \mb{vec}(\bS^{-1})}{\md  \mb{vec}'(\bS^{-1})}\\
&&= \frac{\md \mb{vec}(\bS^{-1})}{\md  \mb{vec}'(\bS^{-1})}  ( \bA\bS^{-1}\otimes \bI_p)+\frac{\md ( \bA\bS^{-1}\otimes \bI_p)}{\md  \mb{vec}'(\bS^{-1}) } \mb{vec}(\bS^{-1})-\bI_{p^2}\\
&&=( \bA\bS^{-1}\otimes \bI_p)+\frac{\md \mb{vec}(\bA\bS^{-1})}{\md  \mb{vec}'(\bS^{-1})} (\bI_p \otimes \bS^{-1})-\bI_{p^2}\\
&&=( \bA\bS^{-1}\otimes \bI_p)+(\bI_p \otimes \bA)(\bI_p \otimes \bS^{-1})-\bI_{p^2}\\
&&= (\bA\bS^{-1} \otimes \bI_p+\bI_p\otimes \bA\bS^{-1}-\bI_{p^2}).
\end{eqnarray*}
and 
\begin{equation*}
\frac{\md \mb{vec}(\bS^{-1})}{\md \mb{vec}'(\bS)} =- (\bS^{-1}\otimes \bI_p)( \bI_p \otimes \bS^{-1} )=-(\bS^{-1} \otimes \bS^{-1} ),
\end{equation*}
where the formulas  
$  (B' \otimes C) \mb{vec}(D)=\mb{vec}(CDB)  $ and  $ (B\otimes C)  (D\otimes E) = (BD\otimes CE) $
are repeatedly used.

\subsection{  Proofs of  Lemma \ref{CLT}  \label{A2} }

First, the result of (\ref{04mean1}) and (\ref{04var1})  is corresponding to the ones in \cite{BS04}  with the 4th moment equal to 3.  Obviously,  the mean in  (\ref{04mean1}) is formed under the condition that  the matrix $T$ in  \cite{BS04} is identity, and its LSD is $H(t)=I_{[1,\infty)}(t)$ according to the assumptions in Lemma \ref{CLT}. Next, If we drop the condition on the 4th moment, it will be found that each of 
the  (4.10) and (2.7) in   \cite{BS04} should be  plused an additional  item by their (1.15)  
\[-\beta q b_p^2(z) \cdot \mb{E} \left( e'_1T^{\frac1{2}}D^{-1}T^{\frac1{2}}e_1 \cdot  e'_1T^{\frac1{2}}D^{-1}(\underline{m}(z)T+\bI)^{-1}T^{\frac1{2}}e_1\right)\]
and
\[\frac{\beta b_p(z_1)b_p(z_2)}{n^2}\sum\limits_{j=1}^n\sum\limits_{i=1}^p e'_i T^{\frac1{2}} \mb{E}_j (D_j^{-1}(z_1)) T^{\frac1{2}} e_i 
\cdot   e'_i T^{\frac1{2}} \mb{E}_j (D_j^{-1}(z_2)) T^{\frac1{2}} e_i \]
respectively, where
\begin{eqnarray*}
&&e_i=(\underbrace{0\cdots,0,}_{i-1}1,0,\cdots,0)' ; \quad
D(z)=T^{\frac1{2}}{ \mS_n }T^{\frac1{2}}-z\bI;  \\
&& D_j(z)=D(z)-r_jr_j^*; \quad r_j=\frac{1}{\sqrt{n}}T^{\frac1{2}}\bxi_{\cdot j}\quad
 b_p(z)=\frac{1}{1+n^{-1}\mb{E tr}T D_1^{-1}}
\end{eqnarray*}
and $\mb{E}_j$ is the conditional expectation given $r_1,\cdots, r_j$ for $j=1,\cdots, n$.

According to the Lemma 6.2 in   \cite{Zheng},  if the 4th moment is arbitrary  finite number, the mean function of $M(z)$ in Lemma 1.1 in   \cite{BS04} should be added 

\[
\frac{\beta q \underline{m}^3(z)\cdot \int \frac{t}{1+t\underline{m}(z)} \md H(t)\cdot \int \frac{1}{(1+t\underline{m}(z))^2} \md H(t)}
{1-q \int \frac{t^2\underline{m}^2(z)}{(1+t\underline{m}(z))^2} \md H(t)}
\]
which is the  limit of 
\[\frac{\beta q \underline{m}(z) b_p^2(z)\cdot \mb{E} \left( e'_1T^{\frac1{2}}D^{-1}T^{\frac1{2}}e_1 \cdot  e'_1T^{\frac1{2}}D^{-1}(\underline{m}(z)T+\bI)^{-1}T^{\frac1{2}}e_1\right)}
{1-q \int \frac{t^2\underline{m}^2(z)}{(1+t\underline{m}(z))^2} \md H(t)} \]
ever dropped in (4.10) and (4.12) of  \cite{BS04}. Similarly, the covariance function   of $M(z)$ should include the additional item 
\[\beta q \cdot \int \frac{t\underline{m}'(z_1)}{(1+t\underline{m}(z_1))^2} \md H(t)\cdot  \int \frac{t\underline{m}'(z_2)}{(1+t\underline{m}(z_2))^2} \md H(t)\]
which is the  limit of 
\[\frac{\partial^2}{\partial z_1 \partial z_2}  \left(\frac{\beta b_p(z_1)b_p(z_2)}{n^2}\sum\limits_{j=1}^n\sum\limits_{i=1}^p e'_i T^{\frac1{2}} \mb{E}_j (D_j^{-1}(z_1)) T^{\frac1{2}} e_i 
\cdot   e'_i T^{\frac1{2}} \mb{E}_j (D_j^{-1}(z_2)) T^{\frac1{2}} e_i \right)\]
ever dropped in (2.7) of   \cite{BS04}. 
Then  by their (1.14), the added mean function of  $G_n(f_j)$ should be
\begin{equation}
-\frac{\beta q }{2 \pi i} \oint f_j(z) \frac{ \underline{m}^3(z)\cdot \int \frac{t}{1+t\underline{m}(z)} \md H(t)\cdot \int \frac{1}{(1+t\underline{m}(z))^2} \md H(t)}
{1-q \int \frac{t^2\underline{m}^2(z)}{(1+t\underline{m}(z))^2} \md H(t)} \md z
\label{meanplus}
\end{equation}
and the covariance function   of $G_n(f_j)$  should plus
\begin{equation}
-\frac{\beta q}{4\pi^2}\oint\oint f_j(z_1)f_\ell(z_2)\int \frac{t\underline{m}'(z_1)}{(1+t\underline{m}(z_1))^2} \md H(t)\cdot  \int \frac{t\underline{m}'(z_2)}{(1+t\underline{m}(z_2))^2} \md H(t) ~\md z_1 \md z_2
\label{varplus}
\end{equation}

Put the condition $H(t)=I_{[1,\infty)}(t)$  assuming in Lemma \ref{CLT}   into the equation (\ref{meanplus}) and (\ref{varplus}), then we have 
the additional mean function
\[
-\frac{\beta q }{2 \pi i} \oint f_j(z)\frac{\underline{m}^3(z)}{(1+\underline{m}(z))[(1-q)\underline{m}^2(z)+2\underline{m}(z)+1]} \md z,
\]
and  the added covariance
function
\[
-\frac{\beta q}{4\pi^2}\oint\oint\frac{f_j(z_1)f_\ell(z_2)}{(1+\underline{m}(z_1))^2(1+\underline{m}(z_2))^2}
d\underline{m}(z_1)d\underline{m}(z_2), \]
where  $ j,\ell \in \{1, \cdots,
k\}$.

\subsection{ Proofs of limiting schemes for the correction to Rao's score test   \label{A3}}

\begin{itemize}

%\vskip 0.1in
\item {\bf  Calculation of  $F^{q_n} (g)$ in (\ref{limitRST}).}
%\vskip 0.1in

Because $F^{q_n}(g)=\int_{-\infty}^{\infty}g(x) \mb{d} F^{q_n}(x)$,  where $F^{q_n}(x)$  is the Mar\v{c}enko-Pastur law of  the matrix $\mS$ with index $q_n$,
the density  is 
\[
 p^{q_n}(x)=
\left\{
\begin{array}{cc}
   \displaystyle\frac{1}{2\pi x q_n}\sqrt{(b_n-x)(x-a_n)}, &  \text{if}   \quad a_n \leq x \leq b_n,    \\
 0, &  \text{otherwise},   
\end{array}
\right.
\]
and has a point mass $1-\displaystyle\frac{1}{q_n} $ at the origin  if $q_n>1$, 
where $a_n=(1-\sqrt{q_n})^2$ and $b_n=(1+\sqrt{q_n})^2$.
(See  in   \cite{BS10}).
According the the definition,
 the supporting set of MP-law is  $x \in [0,4]$ if $q_n=1$.  But  it is unreasonable that $x$ lies on the denominator  if $x=0$ by the expression of the density.
So  we exclude the case of $q_n=1$  and consider the following integral first,
\begin{equation*}
 F^{q_n}(g)=\int_{-\infty}^{\infty}\frac{(x-1)^2}{2\pi x q_n}\sqrt{(b_n-x)(x-a_n)}\md x
\end{equation*}
Make a substitution $x=1+q_n-2\sqrt{q_n}\cos \theta$,  where $0 \leq \theta \leq \pi $,  then 
\begin{eqnarray*}
  F^{q_n}(g)&=&\int^{b_n}_{a_n}\frac {(x-1)^2}{2\pi
    xq_n}\sqrt{(b_n-x)(x-a_n)} \md x\\
  &=&\frac{2}{\pi
    }\int_0^{\pi}\frac{(q_n-2\sqrt{q_n}\cos\theta)^2}
    {1+q_n-2\sqrt{q_n}\cos\theta}\sqrt{\sin^2\theta}\sin\theta \md\theta\\
  &=&\frac{1}{\pi
    }\int_0^{2\pi}\frac{(q_n-2\sqrt{q_n}\cos\theta)^2\sin^2\theta }
    {1+q_n-2\sqrt{q_n}\cos\theta} \md\theta\\
\end{eqnarray*}
Let $x=1+q_n-2\sqrt{q_n}\cos\theta=-2\sqrt{q_n}(\cos\theta+d_0)$, where $d_0=-\displaystyle\frac{1+q_n}{2\sqrt{q_n}}$ is a constant.
Thus, the above integral $F^{q_n} (g)$ is obtained by the partition into three  parts as below :
\begin{eqnarray*}
&&\frac{1}{\pi
    }\int_0^{2\pi}\frac{(q_n-2\sqrt{q_n}\cos\theta)^2\sin^2\theta }
    {1+q_n-2\sqrt{q_n}\cos\theta}\md\theta\\
  &&=  \frac{1}{\pi
    }\int_0^{2\pi}\frac{[-2\sqrt{q_n}(\cos\theta+d_0)-1]^2\sin^2\theta }
    {-2\sqrt{q_n}(\cos\theta+d_0)}\md\theta\\
    &&=\frac{1}{\pi
    }\int_0^{2\pi}\left[-2\sqrt{q_n}\cos\theta\sin^2\theta-(2\sqrt{q_n}d_0+2)\sin^2\theta-\frac{\sin^2\theta  }{2\sqrt{q_n}(\cos\theta+d_0)}\right]\md\theta\\
    &&=\frac{1}{\pi}\left[0+(q_n-1)\pi+\frac{\pi}{2q_n}(1+q_n-|1-q_n|)\right]\\
    &&=
\left\{
\begin{array}{cc}
  q_n, &~if~ q_n<1 ;  \\
  q_n-1+1/q_n,&  ~if~ q_n>1,  
\end{array}
\right.  
\end{eqnarray*}
%%%%%%%%%%%%%%%
where the third part is calculated by the  following  integral, which is also used in other calculations.
\begin{eqnarray}
&&\int_{0}^{2\pi}\frac{1}{\cos\theta+d_0}\md \theta=\int_{0}^{2\pi}\frac{1}{\cos^2\frac{\theta}{2}-\sin^2\frac{\theta}{2}+d_0}\md \theta\nonumber\\
&&=\int_{0}^{2\pi}\frac{2\md \frac{\theta}{2}}{(1-\tan^2\frac{\theta}{2}+d_0\sec^2\frac{\theta}{2})\cos^2\frac{\theta}{2}}\nonumber\\
&&=\int_{0}^{2\pi}\frac{2\md \tan\frac{\theta}{2}}{d_0+1+(d_0-1)\tan^2\frac{\theta}{2}}\nonumber\\
&&=\frac{2}{\sqrt{d_0^2-1}}\displaystyle\int_{0}^{2\pi}
\frac{1}{1+\left(\sqrt{\frac{d_0-1}{d_0+1}}\tan\frac{\theta}{2}\right)^2}
\md \left(\sqrt{\frac{d_0-1}{d_0+1}}\tan\frac{\theta}{2}\right)\nonumber\\
&&=\frac{-2\pi}{\sqrt{d_0^2-1}}\label{eqA1chucosd}
\end{eqnarray}
and the third part of the limiting integral  $F^{q_n}(g)$ is 
\begin{eqnarray}
&&\int_0^{2\pi}-\frac{\sin^2\theta}{2\sqrt{q_n}(\cos\theta+d_0)}\md\theta \nonumber\\
&=&-\frac{1}{2\sqrt{q_n}}\int_0^{2\pi}\frac{1-\cos^2\theta  }{\cos\theta+d_0}\md\theta\nonumber\\
&=&-\frac{1}{2\sqrt{q_n}}\int_0^{2\pi} (-\cos\theta+d_0+\frac{1-d_0^2}{\cos\theta+d_0})\md\theta\nonumber\\
&=&\frac{\pi}{2q_n}(1+q_n-|1-q_n|)\label{eqA1sin2chucosd}
\end{eqnarray}
Because  the density corresponding to $F^{q_n}(x)$  has a point mass $1-\displaystyle\frac{1}{q_n} $ at the origin  if $q_n>1$, then the $F^{q_n}(g)$
should be added the term $(1-0)^2\cdot (1-\displaystyle\frac{1}{q_n} ) $ 
if $q_n>1$.
Then we arrive at 
\[F^{q_n}(g)=q_n, \quad \text{if }  \quad  q_n<1 \quad \&  \quad q_n>1.\]

%\vskip 0.1in
\item{ {\bf Calculation of  $\mu (g)$ in (\ref{meanRST}).}}
%\vskip 0.1in

By  (9.12.13) in Bai and Silverstein \cite{BS10}, with $H(t)=\mb{I}_{[1,\infty}(t)$, the first part of the  limiting mean $\mu(g)$ in (\ref{04mean1}) can also be expressed as 
$$
\mu_1(g)=(\kappa-1) \cdot \left(\frac{g\left(a(q)\right)+g\left(b(q)\right)}{4} -
\frac{1}{2\pi}\int_{a(q)}^{b(q)}\frac{g(x)}{\sqrt{4q-(x-1-q)^2}}\md x\right)\\
$$
where $a(q)=(1-\sqrt{q})^2$ and  $b(q)=(1+\sqrt{q})^2$. For $g(x)= (x-1)^2$, make a substitution
$x=1+q-2\sqrt{q}\cos\theta,~ 0\leq \theta \leq \pi$, then
\begin{eqnarray*}
\mu_1(g)&=&(\kappa-1) \bigg(\frac{g\left(a(q)\right)+g\left(b(q)\right)}{4} \\
&& -\frac{1}{2\pi}\int_{0}^{\pi}
\frac{g(1+q-2\sqrt{q}\cos\theta)}{\sqrt{4q-(2\sqrt{q}\cos\theta)^2}}\cdot 2\sqrt{q}\sin\theta \md \theta\bigg)\\
&=&(\kappa-1)  \left(\frac{g\left(a(q)\right)+g\left(b(q)\right)}{4} -\frac{1}{4\pi}\int_{0}^{2\pi}
g(1+q-2\sqrt{q}\cos\theta)\md \theta\right)\\
&=&(\kappa-1)  \left(\frac{4q+q^2}{2}-\frac{1}{4\pi}\int_0^{2\pi}(q^2-4q^{\frac{3}{2}}\cos\theta+4q\cos^2\theta) \md \theta\right)\\
&=&(\kappa-1) q
\end{eqnarray*}
where $\kappa=2$  if the  variables are real, and $\kappa=1$  if the  variables are complex.

 The second part of the  limiting mean $\mu(g)$  is obtained by  (\ref{04mean2}) 
\[ 
\mu_2(g)=- \frac{\beta q }{2 \pi i} \oint (1-z)^2\frac{\underline{m}^3(z)}{(1+\underline{m}(z))[(1-q)\underline{m}^2(z)+2\underline{m}(z)+1]} \md z,
\]
For $z \in \mathbb{C}^{+}$, recall  the equation (9.12.12) given in \cite{BS04}
\[
 z=-\frac{1}{\underline{m}(z)}+\frac{q}{1+\underline{m}(z)}.
\]
Denote $\underline{m}(z)$  as $ m$  for simplicity, it is easily obtained that 
\[(1-z)^2=\frac{[m^2-(q-2)m+1]^2}{m^2(1+m)^2}  \quad \quad \md z = \frac{(1-q)m^2+2m+1}{m^2(1+m)^2} \md m\]
 
then we have 
\[ 
\mu_2(g)= -\frac{\beta q }{2 \pi i} \oint \frac{[m^2-(q-2)m+1]^2}{m(1+m)^5} \md m,
\]
By solving for $m$ from (9.12.12)    in \cite{BS10},  we get the contour for the integral above should enclose the interval 
\[
\left[\min(-\frac{1}{1-\sqrt{q}},  -\frac{1}{1+\sqrt{q}}),  \max(-\frac{1}{1-\sqrt{q}},  -\frac{1}{1+\sqrt{q}})\right ].
\]
Therefore,  -1 is the residue if $q\leq 1$ and 0 is the residue if $ q>1$.  and the integral is calculated as 
\[\mu_2(g)=\beta q, \]
 which is the same result for both the cases of $q \leq 1$
 and  $q>1$.
 
 Finally, we  obtained 
 \[\mu(g)=(\kappa-1)q+\beta q .\]  
 
% \vskip 0.1in
\item {\bf Calculation of  $\upsilon (g)$ in (\ref{varRST}).}
%\vskip 0.1in

By Lemma \ref{CLT},  the first part of  limiting variance $\upsilon (g)$ in (\ref{04var1}) is 
$$\upsilon_1(g)=-\frac{\kappa}{4\pi^2}\oint\oint\frac{g(z_1)g(z_2)}{(\underline{m}(z_1)-\underline{m}(z_2))^2}
\md\underline{m}(z_1)\md\underline{m}(z_2)$$ and
\begin{eqnarray*}
g(z_1)g(z_2)&=&(z_1-1)^2(z_2-1)^2\\
&=&1-2z_1-2z_2+z_1^2+z_2^2+4z_1z_2-2z^2_1z_2-2z_1z^2_2+z^2_1z^2_2.
\end{eqnarray*} 
%%%%%%%%%%
Let \textbf{1} denote constant function which equals to 1,  It is obvious that  $\upsilon(\textbf{1},\textbf{1})=0$.
Denoting $\underline{m}(z_i)$ = $m_i,~~i=1,2.$
As mentioned above,
 for fixed $m_2$,~we
have  a contour enclosed   
 -1, but not 0 when $0\leq q \leq 1$, whereas it enclosed  0, but not  
  -1 when $ q>1$. 

On one hand,  we consider the case of $0 \leq q \leq1$.  Because
 \begin{eqnarray*}
   &&\displaystyle {\int \frac{z_1}
     {(m_1-m_2)^2}\md m_1}\\
  % &=& \displaystyle{\int
     %\frac{-(1+ m_1)+y m_1}{m_1(1+ m_1)}
   %  ( m_2-m_1)^{-2}dm_1}\\
%   &&=\displaystyle{q\int (\frac{1}{1+m_{1}}+\frac{1-q}{q}) \cdot
%     [1-(1+m_{1})]^{-1}
%     \cdot (m_{2}+1)^{-2} \cdot (1-\frac{m_{1}+1}{m_{2}+1})^{-2}
%     \md m_{1}}\\
   &&=\displaystyle {q\int (\frac{1}{1+m_{1}}+\frac{1-q}{q}) 
     \sum\limits^\infty_{j=0}(1+m_1)^j (m_2+1)^{-2}
     \sum\limits^\infty_{\ell=1}\ell(\frac{m_{1}+1}{m_{2}+1})^{\ell-1}\md m_1}\\
   &&=\displaystyle{2\pi i\cdot  \frac{q}{(1+m_2)^{2}}}.
 \end{eqnarray*}
  and 
  \begin{eqnarray}
  &&\displaystyle {\oint \frac{z_2^2}
     {(m_1-m_2)^2}\md m_1}\nonumber\\
     &&=\oint2\left(-\frac{1}{m_2}+\frac{q}{1+m_2}\right)\left(\frac{1}{m^2_2}-\frac{q}{(1+m_2)^2}\right) \frac{1}{m_2-m_1}\md m_2\nonumber\\
   %  &=&2\oint\left[\frac{q}{m_2(m_2-m_1)(m_2+1)^2} +\frac{q}{m_2^2(m_2-m_1)(m_2+1)}-\frac{q^2}{(m_2-m_1)(1+m_2)^3}\right]\md m_2\\
     && = 4\pi i \left[\frac{q}{(1+m_1)^2}+\frac{q^2}{(1+m_1)^3} \right]. \label{eqvarz2}
\end{eqnarray}
So $\upsilon(z_1^2-2z_1, ~\textbf{1})=0$.
Similarly,
$\upsilon(\textbf{1}, ~ z_2^2-2z_2)=0$.\\
Therefore, there are only four parts left, i.e. $z^2_1z^2_2-2z^2_1z_2-2z_1z^2_2+4z_1z_2$.   
Further,
\begin{eqnarray*}
&&\upsilon(z_1,z_2)\\
&&=\displaystyle{\frac {\kappa q^2}{2\pi
i}\int\frac{1}{(m_2+1)^{2}}(\displaystyle\frac{1}{1+m_{2}}+\displaystyle\frac{1-q}{q})
\displaystyle\sum\limits^\infty_{j=0}(1+m_2)^j\md m_2}\\
&&=\kappa q\\[5mm]%\end{eqnarray*}
%\begin{eqnarray*}
&&\upsilon(z^2_1,z_2)\\
&&=-\frac{\kappa}{4\pi^2}\oint\oint\frac{z^2_1z_2}{(m_1-m_2)^2}
\md m_1 \md m_2\\
&&=\displaystyle{\frac {\kappa q}{2\pi
i}}\oint\frac{\left(-\frac{1}{m_1}
       +\frac{q}{1+m_1}\right)^2}{(1+m_1)^2} \md  m_1\\
       &&=\frac {\kappa q}{2\pi
i}\oint2\left(-\frac{1}{m_1}+\frac{q}{1+m_1}\right)\left(\frac{1}{m^2_1}-\frac{q}{(1+m_1)^2}\right) \frac{1}{1+m_1}\md m_1\\
%&&=\frac {\kappa q}{\pi i}\oint\left[-\frac{1}{m_1^3(1+m_1)} +\frac{q}{m_1(1+m_1)^3}+\frac{q}{m^2_1(1+m_1)^2}-\frac{q^2}{(1+m_1)^4}\right]\md m_1\\
&&= 2\kappa q(1+q).
\end{eqnarray*}
Similarly, $
\upsilon(z_1,z^2_2)= 2\kappa q(1+q)$.
For the last part $\upsilon(z^2_1,z^2_2)$, the integral is calculated  by eq.(\ref{eqvarz2}) as below.
\begin{eqnarray*}
&&\upsilon(z^2_1,z^2_2)=-\frac{\kappa}{4\pi^2}\oint\oint\frac{z^2_1z^2_2}{(m_1-m_2)^2}
\md m_1 \md m_2\\
&&=\displaystyle{\frac {\kappa q}{\pi
i}}\oint\left(-\frac{1}{m_1}
       +\frac{q}{1+m_1}\right)^2\left[\frac{1}{(1+m_1)^2} +\frac{q}{(1+m_1)^3}\right] \md  m_1\\
&&=\frac { \kappa q}{\pi
i}\oint2\left(-\frac{1}{m_1}+\frac{q}{1+m_1}\right)\left(\frac{1}{m^2_1}-\frac{q}{(1+m_1)^2}\right) \frac{1}{1+m_1}\md m_1\\
&&+\frac {\kappa q^2}{\pi
i}\oint2\left(-\frac{1}{m_1}+\frac{q}{1+m_1}\right)\left(\frac{1}{m^2_1}-\frac{q}{(1+m_1)^2}\right)\left[ \frac{1}{2(1+m_1)^2}\right]\md m_1\\
&&=\kappa (4q+10q^2+4q^3)
\end{eqnarray*}

 Finally,
we obtain 
\begin{eqnarray*}
\upsilon_1(g)&=&\upsilon(z^2_1, z^2_2)-2\upsilon(z^2_1, z_2)-2\upsilon(z_1, z^2_2)+4\upsilon(z_1, z_2)\\
&=&\kappa (4q+10q^2+4q^3)
- 8\kappa q(1+q)+4\kappa q\\
&=&2\kappa q^2(1+2q).
\end{eqnarray*}
when $0 \leq q\leq 1$.

On the other hand,  Similar calculations  are conducted  for  the case of $q>1$. It is found that the result is the same as the one  of the case above  only with the residues are changed from  -1 to 0. So for all the cases of $q$, we arrive at 
\[\upsilon_1(g)=2\kappa q^2(1+2q)\]

For the second  part of    $\upsilon (g)$ in (\ref{04var2}), we have 
\[
\upsilon_2(g)=-\frac{\beta q}{4\pi^2}\oint\oint\frac{g(z_1)g(z_2)}{(1+m_1)^2(1+m_2)^2}
\md m_1\md m_2.
\]
Furthermore,
\begin{eqnarray*}
\oint\frac{g(z_1)}{(1+m_1)^2}\md m_1
=\oint \frac{[m_1^2-(q-2)m_1+1]^2}{m_1^2(1+m_1)^4} \md m_1
=-4\pi i q
\end{eqnarray*}
Since the  contour contains 
 -1 as a residue if  $0\leq q \leq 1$, and   enclose 0 as a residue  for the other case $ q>1$. 
By the  calculations of the both cases, it will be found that the results are all the same although the residues are different.
Thus we get
\[\upsilon_2(g)=-\frac{\beta q}{4\pi^2} \cdot (-4\pi i q) \cdot(-4\pi i q) =4\beta q^3.\]

 Finally, we  obtained 
 \[\upsilon(g)=2\kappa q^2(1+2q)+4\beta q^3.\]  
\end{itemize}

\section*{Acknowledgement}
The author thanks the reviewers for their helpful comments and suggestions to make an improvement of this article. This research was supported by the National Natural Science Foundation of China 11471140.

\section*{References}
\bibliography{Corrections to Rao's Score tests by DandanJiang.bbl}

\end{document}